\documentclass[a4paper, USenglish, numberwithinsect, cleveref, autoref]{lipics-v2019} 
\usepackage{hyperref}
\usepackage{amsmath,amssymb,amsthm}
\usepackage{xcolor}
\usepackage{graphicx}
\usepackage{lipsum,booktabs}

\let\footnotesize\small

\newcommand{\E}{\mathbf{E}}

\newcommand{\bigA}{{\cal A}}

\newcommand{\bigD}{{\cal D}}

\def\paragraph#1{\vspace{0.25em}\noindent {\bf #1}}
\renewcommand{\thenote}{\thesection.\arabic{note}}

\newcommand{\mhnote}[1]{{\color{blue}\refstepcounter{note}$\ll${\bf Magnus's Comment \thenote: }{#1}$\gg$\marginpar{\tiny\bf MMH \thenote}}}

\newif\iflong
\longtrue

\newif\ifshort
\shortfalse

\title{The Capacity of Smartphone Peer-to-Peer Networks}

\author{Michael Dinitz}
{Johns Hopkins University, Baltimore, MD, United States}
{mdinitz@cs.jhu.edu}
{}
{Supported in part by NSF award CCF-1535887}

\author{Magn\'us M. Halld\'orsson}
{Reykjav\'ik University, Iceland}
{mmh@ru.is}
{} 
{Supported in part by Icelandic Research Fund grant 174484} 

\author{Calvin Newport}
{Georgetown University, United States}
{cnewport@cs.georgetown.edu}
{} 
{Supported in part by NSF award CCF-1733842} 

\author{Alex Weaver}
{Georgetown University, United States}
{aweaver@cs.georgetown.edu}
{} 
{Supported in part by NSF award CCF-1733842} 

\authorrunning{M. Dinitz, M.\,M. Halld\'orsson, C. Newport, and A. Weaver}

\Copyright{Michael Dinitz, Magn\'us M. Halld\'orsson, Calvin Newport, and Alex Weaver}



\nolinenumbers 

\category{}

\ifshort
\relatedversion{A full version of this paper is available at \url{https://arxiv.org/abs/1908.01894}.}
\fi
\supplement{}

\funding{}

\acknowledgements{}

\iflong
\hideLIPIcs
\fi

\EventEditors{Jukka Suomela}
\EventNoEds{1}
\EventLongTitle{33rd International Symposium on Distributed Computing (DISC 2019)}
\EventShortTitle{DISC 2019}
\EventAcronym{DISC}
\EventYear{2019}
\EventDate{October 14--18, 2019}
\EventLocation{Budapest, Hungary}
\EventLogo{}
\SeriesVolume{146}
\ArticleNo{14}

\begin{document}

\maketitle

\begin{abstract}

We study three capacity problems in the mobile telephone model, a network abstraction that models the peer-to-peer communication capabilities implemented in most commodity smartphone operating systems.  The \emph{capacity} of a network expresses how much sustained throughput can be maintained for a set of communication demands, and  is therefore a fundamental bound on the usefulness of a network.  Because of this importance, wireless network capacity has been active area of research for the last two decades.  

The three capacity problems that we study differ in the structure of the communication demands.  The first problem is pairwise capacity, where the demands are (source, destination) pairs.  Pairwise capacity is one of the most classical definitions, as it was analyzed in the seminal paper of Gupta and Kumar on wireless network capacity.  The second problem we study is broadcast capacity, in which a single source must deliver packets to all other nodes in the network.  Finally, we turn our attention to all-to-all capacity, in which all nodes must deliver packets to all other nodes.  In all three of these problems we characterize the optimal achievable throughput for any given network, and design algorithms which asymptotically match this performance.  We also study these problems in networks generated randomly by a process introduced by Gupta and Kumar, and fully characterize their achievable throughput.  

Interestingly, the techniques that we develop for all-to-all capacity also allow us to design a one-shot gossip algorithm that runs within a polylogarithmic factor of optimal in every graph. This largely resolves an open question from previous work on the one-shot gossip problem in this model. 

\end{abstract}

\newpage




\section{Introduction}
\label{sec:intro}

In this paper, we study the classical capacity problem in the mobile telephone model:
an abstraction that models the peer-to-peer communication capabilities implemented in most
commodity smartphone operating systems.
The capacity of a network expresses how much sustained throughput can be maintained for a set of communication demands.
We focus on three variations of the problem:
{\em pairwise} capacity, in which nodes are divided into pairwise packet flows,
{\em broadcast} capacity, in which a single source delivers packets to the whole network, 
and {\em all-to-all}
capacity, in which all nodes deliver packets to the whole network.

For each variation we prove limits on the achievable throughput and analyze
algorithms that match (or nearly match) these bounds. 
We study these results in both {\em arbitrary} networks and {\em random} networks generated with
the process introduced by Gupta and Kumar in their seminal paper on wireless network capacity~\cite{gupta:2000}.
Finally, we deploy our new techniques to largely resolve an open question from~\cite{newport:2017}
regarding optimal one-shot gossip in the mobile telephone model.
Below we summarize the problems we study and the results we prove,
interleaving the relevant related work.

\paragraph{The Mobile Telephone Model.}
The \emph{mobile telephone model} (MTM),
introduced by Ghaffari and Newport \cite{ghaffari:2016},
 modifies the well-studied 
\emph{telephone model} of wired peer-to-peer networks (e.g., \cite{frieze1985shortest,giakkoupis2011tight,chierichetti2010rumour,giakkoupis2012rumor,fountoulakis2010rumor,giakkoupis2014tight})
to better capture the dynamics of standard smartphone peer-to-peer libraries.
It is inspired, in particular, by the specific interfaces 
provided by Apple's Multipeer Connectivity Framework~\cite{multipeer}.

In this model,
the network is modeled as an undirected graph $G=(V,E)$,
where the nodes in $V$ correspond to smartphones,
and an edge $\{u,v\}\in E$ indicates the devices
corresponding to $u$ and $v$ are close enough to enable a direct peer-to-peer radio link.
Time proceeds in synchronous rounds.
As in the original telephone model,
in each round, each node can either attempt to initiate a connection (e.g., place a telephone call)
with at most one of its neighbors, or wait to receive connection attempts.
Unlike the original model, however,
a waiting node can accept at most one incoming connection attempt.
This difference is consequential,
as many of the celebrated results of the original telephone model depend
on the nodes' ability to accept an unbounded number of incoming connections (see~\cite{ghaffari:2016,daum:2016} for more discussion).\footnote{This behavior
is particularly evident in studying PUSH-PULL rumor spreading in the telephone model in a star network topology.
This simple strategy performs well in this network due to the ability of the points of the star
to simultaneously pull the rumor from the center. In the mobile telephone model, by contrast,
any rumor spreading strategy would be fundamentally slower due to the necessity of the center to connect
to the points one by one.}
This restriction is motivated by the reality that standard smartphone peer-to-peer libraries limit
the number of concurrent connections per device to a small constant (e.g., for Multipeer this limit is $8$).
Once connected, a pair of nodes can participate in a bounded amount of reliable communication (e.g.,
transfer a constant number of packets/rumors/tokens).

Finally, the mobile telephone model also allows each node to broadcast a small $O(\log{n})$-bit advertisement 
to its neighbors at the start of each round before the connection decisions are made.
Most existing smartphone peer-to-peer libraries implement this {\em scan-and-connect} architecture.
Notice, the mobile telephone model is harder than the original telephone model due to its connection restrictions, but also easier due to the presence of advertisements. The results is that the two settings are formally incomparable: each requires its own strategies for solving key problems.

In recent years, 
several standard one-shot peer-to-peer problems have been studied in
the MTM, including rumor spreading \cite{ghaffari:2016}, load balancing \cite{dinitz:2017}, leader election \cite{newport:2017}, and gossip \cite{newport:2017,newport:2019}.
This paper is the first to study ongoing communication in this setting.

\paragraph{The Capacity Problem.}
Capacity problems are parameterized with a network topology $G=(V,E)$,
and a flow set $F$ made up of pairs of the form $(s,R)$ (each of which is a \emph{flow}), where $s\in V$ indicates
a source (sometimes called a sender), and $R\subset V$ indicates a set of destinations (receivers).
For each flow $(s,R) \in F$, source $s$ is tasked with routing an infinite sequence of packets
to destinations in $R$.
The throughput achieved by a given destination for a particular flow is the average number of packets it receives from that flow per round in the limit,
and the overall throughput is the smallest throughput over all the destinations in all flows (see Section~\ref{sec:problem} for formal definitions).
We study three different capacity problems,
each defined by the different constraints they place on the flow set $F$.

\paragraph{Results: Pairwise Capacity.}
The pairwise capacity problem divides nodes into source and destination pairs in $F$, i.e., the given flows are between pairs of nodes rather than from a source to a general destination set.
We begin with pairwise capacity as it was the primary focus of Gupta and Kumar's seminal paper
on the capacity of the protocol and physical wireless network models~\cite{gupta:2000}.
They argued that it provides a useful assessment of a network's ability to handle concurrent communication.

We begin in Section~\ref{sec:mcf} by tackling the following fundamental problem: given an arbitrary connected
network topology graph $G=(V,E)$ and a {flow set} $F$ that divides the nodes in $V$ into sender
and receiver pairs, is it possible to efficiently calculate a packet routing schedule that approximates
the optimal achievable throughput? 
We answer this question in the affirmative by establishing a novel connection between pairwise
capacity and the classical concurrent multi-commodity flow (MCF) problem.
To do so, we first transform a given $G$ and $F$ into an instance of the MCF problem.
We then apply an existing MCF approximation algorithm to generate a fractional flow that achieves
a good approximation of the optimal flow in the network.
Finally, we apply a novel rounding procedure to transform the fractional flow into a schedule.
We prove that this resulting schedule provides a constant approximation of the optimal achievable throughput.

Inspired by Gupta and Kumar~\cite{gupta:2000},
in Section~\ref{sec:gk:random}
we turn our attention to networks and flow pairings that are randomly
generated using the process introduced in~\cite{gupta:2000}.
This process is parameterized with a network size $n\geq 2$ and communication radius $r>0$.
It randomly places the $n$ nodes in a unit square and adds an edge between any
pair of nodes within distance $r$. The source and destination pairs are also randomly generated.

For every given size $n$, we identify a {\em connectivity threshold} value $r_c(n) = \Theta(\sqrt{\log{n}/n})$,
such that for any radius $r \leq r_c(n)$, with constant probability the network generated by the above process for $n$ and $r$ includes a source with no path to its destination---trivializing the optimal achievable throughput to $0$.
We then prove that for every radius $r$ that is at least a sufficiently large constant factor larger than
the threshold, there is a tight bound of $\Theta(r)$ on the optimal achievable throughput.
These results fully characterize our algorithm from Section~\ref{sec:mcf} in
randomly generated networks.

\paragraph{Results: Broadcast Capacity.}
Broadcast capacity is another natural communication problem in which a single {\em source} node 
is provided an infinite sequence of packets to deliver to all other nodes in the network.
Solutions to this problem would be useful, for example,
in a scenario where a large file is being distributed in a peer-to-peer network
of smartphone users in a setting without infrastructure.
In Section~\ref{sec:bcast:upper} we study the optimal achievable throughput for this problem
in arbitrary connected graphs.
To do so, we connect the scheduling of broadcast packets to existing results on {\em graph toughness},
a metric that captures a graph's resilience to disconnection
that was introduced by Chv\'{a}tal~\cite{toughness} in
the context of studying Hamiltonian paths.

In more detail, a graph $G$ has a $k$-tree if there exists a spanning tree of $G$ with maximum degree $k$.
Let $d(G)$ be the smallest $k$ such that $G$ has a $k$-tree. This tree is also called a
 minimum degree spanning tree (MDST) of $G$.
Building on a result of Win~\cite{win:1989} that relates $k$-trees to toughness,
we prove that for any given $G$ with $d(G) > 3$,
there exists a subset $S$ of nodes such that removing $S$ from $G$ partitions
the graph into at least $(d(G) - 2)|S|$ connected components.

As we formalize in Section~\ref{sec:bcast:upper}, 
because each node in $S$ can connect to at most one component per round (due to the 
connection restrictions of the 
mobile telephone model), $\Omega(d(G))$ rounds are required to spread each packet
to all components, implying that no schedule achieves throughput better than $O(1/d(G))$.

In Section~\ref{sec:bcast:lower}, we prove this bound tight by exhibiting a matching
 algorithm. 
The algorithm begins by constructing a $k$-tree $T$ with $k\in \Theta(d(G))$
using existing techniques; e.g.,~\cite{fr94,dinitz:2019}.
It then edge colors $T$ and uses the colors as the foundation for a TDMA schedule of length $\Theta(k)$
that allows nodes to simulate the more powerful CONGEST model in which each node can connect with every neighbor 
in a round. In the CONGEST model, a basic pipelined broadcast provides constant throughput.
When combined with the simulation cost the achieved throughput is an asymptotically optimal $\Omega(1/d(G))$.

It is straightforward for a centralized algorithm to calculate this schedule in polynomial time,
but in some cases a pre-computation of this type might be impractical,
or require too high of a setup cost.\footnote{In the mobile telephone model, 
all nodes can learn the entire network topology in $O(n^2)$
rounds and then run a centralized algorithm locally to determine their routing behavior.
Though this setup cost is averaged out when calculating throughput in the limit, 
it might be desirable to minimize it in practice.}
%
%
%
%
With this in mind,
we also provide a {distributed} version of this algorithm
that converges
to $\Omega(1/(d(G) + \log{n}))$ throughput in $\tilde{O}(D(T)d(G) + \sqrt{n})$ rounds,
where $D(T)$ is the diameter of the spanning tree and $\tilde{O}$ hides polylog$(n)$ factors.
The algorithm further converges to an optimal $\Omega(1/d(G))$ throughput after no more than $O(n^2)$ total rounds---providing
a trade-off between setup cost and eventual optimality. 

Finally, in Section~\ref{sec:bcast:random},
we study the performance of our algorithm in networks generated randomly
using the Gupta and Kumar process summarized above.
We prove that for any communication radius sufficiently larger than the connectivity threshold,
the network is likely to include an $O(1)$-tree, enabling our algorithms to converge to constant throughput.
This result indicates that in evenly distributed network deployments the mobile telephone model is well-suited
for high performance broadcast.

\paragraph{Results: All-to-All Capacity.}
All-to-all capacity generalizes broadcast capacity such that now {\em every} node is provided
an infinite sequence of packets it must deliver to the entire network.
Solutions to this problem would be useful, for example, in a local multiplayer gaming scenario
in which each player needs to keep track of the evolving status of all other players connected
in a peer-to-peer network.

Clearly, $n$ separate instances of our broadcast algorithm from Section~\ref{sec:bcast:lower},
one for each of the $n$ nodes as the broadcast source,
can be interleaved with a round robin schedule to produce $\Omega(1/(n\cdot d(G)))$ throughput.
In Section~\ref{sec:all}, we draw on the same graph theory connections as before
to prove that this result is tight for all-to-all capacity.
We then provide a less heavy-handed distributed algorithm for 
achieving this throughput. Instead of interleaving $n$ different broadcast instances,
it executes distinct instances of all-to-all gossip, one for each packet number,
using a flood-based strategy on a low degree spanning tree.
Finally,  we apply the random graph analysis
from
Section~\ref{sec:bcast:random} to establish that for sufficiently large communication radius,
with high probability, the randomly generated graph supports $\Omega(1/n)$-throughput, 
which is trivially optimal in the sense that a receiver
can receive at most one new packet per round in our model.

\paragraph{New Results on One-Shot Gossip.}
As we detail in Section~\ref{sec:all:oneshpt},
our results on all-to-all capacity imply new lower and upper bounds on one-shot gossip in the mobile telephone model.
From the lower bound perspective, they imply that gossiping in graph $G$ in the mobile
telephone model requires $\Omega(n\cdot d(G))$ rounds.
From the upper bound perspective,
when we carefully account for the costs of our routing algorithm applied to spreading only a single packet from
each source,
we solve the one-shot problem with high probability in the following number of rounds: 
\[ O((D+\sqrt{n})\text{polylog}(n) + n(d(G) + \log{n}))= \tilde{O}(d(G)\cdot n),\]

\noindent  where $D$ is the diameter of $G$.
This algorithm
is asymptotically optimal in any graph with $d(G) \in \Omega(\log{n})$ and $D \in O(n/\log^x{n})$ (where $x$
is the constant from the polylog in the MDST construction time), which describes a large family of graphs.
For all other graphs the solution is at most a polylog factor slower than optimal.
This is the first known gossip solution to be optimal, or within log factors of optimal, in {\em all} 
graphs, largely answering a challenge presented by~\cite{newport:2017}.

\paragraph{Motivation.} 
Smartphone operating systems include increasingly robust support for opportunistic device-to-device communication
through standards such as Apple's Multipeer Connectivity Framework~\cite{multipeer}, 
Bluetooth LE~\cite{gomez2012overview}, and WiFi Direct~\cite{camps2013device}.
Though the original motivation for these links was to support information transfer among
a small number of nearby phones, researchers are beginning to explore their potential
to enable large-scale peer-to-peer networks.
Recent work, for example, uses smartphone peer-to-peer networking to 
provide disaster response~\cite{suzuki2012soscast,reina2015survey,lu2016networking}, 
circumvent censorship~\cite{firechat}, extend internet access~\cite{aloi2014spontaneous,oghostpot},
support local multiplayer gaming~\cite{mark2015peer} and improve classroom interaction~\cite{holzer2016padoc}.

It remains largely an open question whether or not it will be possible to build large-scale network
systems on top of smartphone peer-to-peer links.
As originally argued by Gupta and Kumar~\cite{gupta:2000},
bounds for capacity problems can help resolve such questions for a given network model
by establishing the limit to their ability to handle ongoing and concurrent communication.
The results in this paper, as well as the novel technical tools developed to prove them, 
can therefore help resolve this critical question concerning this important emerging network setting.



\section{Preliminaries}
Here we define our model, the problem we study, and some useful mathematical tools and definitions. 

\subsection{Model}
\label{sec:model}

The mobile telephone model describes a smartphone peer-to-peer network topology as an undirected graph 
$G=(V,E)$. 
The nodes in $V$ correspond to the smartphone devices,
and an edge $\{u,v\}\in E$ implies that the devices corresponding to $u$ and $v$ are within
range to establish a direct peer-to-peer radio link.
We use $n=|V|$ to indicate the network size.

Executions proceed in synchronous rounds labeled $1,2,...$,
and we assume all nodes start during round $1$.
%
At the beginning of each round,
each node $u\in V$ selects an {\em advertisement} of size at most $O(\log{n})$ bits
to broadcast to its neighbors $N(u)$ in $G$.
After the advertisement broadcasts,
each node $u$ can either  send a connection invitation to at most one neighbor,
or wait to receive invitations.
A node receiving invitations can accept at most one, forming a reliable pairwise connection.
It follows from these constraints that the set of connections in a given round forms a matching.

Once connected, a pair of nodes can perform a bounded amount of reliable communication.
For the capacity problems studied in this paper,
we assume that a pair of connected nodes can transfer at most one packet over the connection in a given round.
We treat these packets as black boxes that can only be delivered in this manner (e.g., you cannot
break a packet into pieces, or attempt to deliver it using advertisement bits).

We assume when running a distributed algorithm in this model that each computational process (also called a {\em node})
is provided a unique ID that can fit into its advertisement and an estimate of the network size.
It is provided no other {\em a priori} information about the network topology,
though any such node can easily learn its local neighborhood in a single round if all nodes advertise their ID.

\subsection{Problem}
\label{sec:problem}

In this paper we measure capacity as the achievable throughput for various combinations of packet flow and network types.
We begin by providing a general definition of {\em throughput} that applies to all settings we study.
This definition makes use of an object we call a {\em flow set}, which is a set $F = \{(s_i, R_i) : 1 \leq i \leq k\}$ (for some $k \geq 1$) where each $s_i \in V$ and $R_i \subseteq V$ (for node set $V$).
For a given flow set $F$, each $(s_i,R_i) \in F$
describes a packet flow of \emph{type $i$}; i.e., source $s_i$ is tasked with sending packets to all the destinations in set $R_i$. We refer to the packets from $s_i$ as \emph{$i$-packets}.

A {\em schedule} for a given $G$ and $F$ describes a movement of packets through the flows defined by $F$. Formally, a schedule is an infinite
sequence of directed matchings, $M_1, M_2, ...$ on $G$, such that
the edges in each $M_t$ are labelled by packets, where we define
a packet as a pair $(i, j)$ with $i \in [|F|]$ and  $j \in \mathbb{N}$ (i.e., $(i,j)$ is the $j$'th packet of type $i$).
We require that the packet labels for a schedule
satisfy the property that if edge $(u,v)$ in $M_t$ is labelled
with packet $p = (i,j)$, then there is a path in $\bigcup_{l <t} M_l$
from $s_i$ to $u$ where all edges on the path are labelled with $p$. (It is easy to see by induction that this corresponds precisely to the intuitive notion of packets moving through a mobile telephone network).  We say that a packet $p$ is \emph{received} by a node $u$ in round $r$ if there is an edge $(v,u) \in M_r$ which is labelled $p$.
A packet $(i,j)$ is \emph{delivered} by round $r$ if
every $x \in R_i$ receives it in some round $t$ with $t \le r$. 

Given a schedule $S$ for a graph $G$ and flow set $F$,
we can define the throughput achieved by 
the slowest rate, indicated in packets per round, at which any
of the flows in $F$ are satisfied in the limit.
Formally:

\begin{definition}
Fix a schedule $S$ defined with respect to
 network topology graph $G=(V,E)$  and flow set $F$.
We say $S$ 
{\em achieves throughput $t$}  with respect to $G$ and $F$,
if  there exists a {\em convergence round} $r_0\geq 1$, such
that for every $r \geq r_0$ and every packet type $i$:
  \[ \left( \frac{del_i(r)}{r} \right)\ge t, \]
where $del_i(r)$ is the largest $j$ such that
for every $l \le j$, packet $(i,l)$ has been delivered by round $r$. 
\end{definition}

The above definition of throughput concerns performance in the limit, since $r_0$ can be arbitrarily large. 
In some cases, though, we might also be concerned with how quickly we achieve this limit.  Our notion of convergence round allows us to quantify this, so we will provide bounds on the convergence round where relevant.  

Many of the results in this paper concern algorithms that produce schedules.
Our centralized algorithms take $G$ and $F$ as input and efficiently
produce a compact description of an infinite schedule (i.e., an infinitely repeatable
finite schedule). Our distributed algorithms assume a computational process running
at each node in $G$, and for each $(s_i,R_i)\in F$, the source $s_i$ is provided an infinite sequence
of packets to deliver to $R_i$.
An execution of such a distributed algorithm might contain communication other than
the {\em flow packets} provided as input; e.g., the algorithm might distributedly (in the mobile telephone model) compute a routing
structure to coordinate efficient packet communication.
However, a unique schedule can be extracted from each such execution by considering
only communication corresponding to the flow packets. \iflong (It is here that we leverage
the model assumption that the set of connections in a given round is a matching and
each connection can send at most one flow packet per round.)  \fi

While our definition of throughput is for schedules and not algorithms, we will say that an algorithm \emph{achieves}  throughput $\alpha$ if it results in a schedule that achieves throughput $\alpha$.

In the sections that follow, we consider three different types of capacity: {\em pairwise}, {\em broadcast},
and {\em all-to-all}. Each capacity type can be formalized as a set of constraints on the allowable flow sets.
For each capacity type we study achievable throughput with respect to both {\em arbitrary} and {\em random} network topology graphs.
In the arbitrary case, the only constraints on the graph is that it is connected.
For the random case, we must describe a process for randomly generating the graph.
To do so, we use the approach introduced for this purpose by Gupta and Kumar~\cite{gupta:2000}: randomly place nodes in a unit square, and then add an edge between all pairs within some fixed radius. 
Formally:

\begin{definition}
For a given real value {\em radius} $r$, $0 < r \leq 1$,  and network size $n\geq 1$, 
the {\em GK$(n,r)$ network generation process} randomly generates a network topology $G=(V,E)$ as follows:
\begin{enumerate}
	\item Let $V=\{u_1,u_2, ..., u_{n}\}$.
	Place each of the $n$ nodes in $V$ uniformly at random in a unit square in the Euclidean plane. 
	
	\item Let $E =\{ (u_i, u_j) : d(u_i, u_j) \leq r\}$, where $d$ is the Euclidean distance metric.
\end{enumerate}
\end{definition}

We will use the notation $G \sim GK(n,r)$ to denote that $G$ is a random graph generated by the $GK(n,r)$ process.  When studying a specific definition of capacity with respect to a network randomly generated with the $GK$
process, it is necessary to specify how the flow set is generated. Because these details differ for each
of the three capacity definitions, we defer their discussion to their relevant sections.

\subsection{Mathematical Preliminaries}
\label{sec:prelim}

\iflong
\paragraph{Probabilistic Preliminaries.}
Several proofs will make use of the following Chernoff bound form:

\begin{theorem}
Suppose $X_1, ..., X_k$ are independent random variables. Let $X=\sum_{i=1}^k X_i$ and $\mu = \E(X)$.
Then,
\begin{itemize}
    \item For $0 < \delta < 1$: $\Pr[X \leq (1-\delta)\mu)] \leq \exp{\left(\frac{-\delta^2 \mu}{2}\right)}$
    \item For $\delta > 0$: $\Pr[X \geq (1+\delta)\mu)] \leq \exp{\left(\frac{-\delta^2 \mu}{2+\delta}\right)}$
\end{itemize}
\label{thm:chernoff}
\end{theorem}
\fi

\iflong \paragraph{Graph Theory Preliminaries.} \fi
We begin with some basic definitions.
Fix some connected undirected graph $G=(V,E)$.
We define $c(G)$ to be the number of components in $G$. In a slight abuse of notation, 
we define $G\setminus S$, for $S \subseteq V$, to be the graph defined when we remove from $G$ the nodes in $S$ and their adjacent edges.
For a fixed integer $k>1$, we say $G$ {\em has a $k$-tree} if there exists a spanning tree in $G$ with maximum degree $k$. 
Finally, let $d(G)$ be the smallest $k$ such that $G$ has a $k$-tree. That is, $d(G)$ describes the maximum degree of the
{\em minimum degree spanning tree} (MDST) in $G$.


\iflong
Several of our capacity results build on a graph metric called {\em toughness}, introduced by Chv\'{a}tal~\cite{toughness} in
the context of studying Hamiltonian paths. It is defined as follows:

\begin{definition}
An undirected graph $G = (V,E)$ has toughness $t(G)$ if $t(G)$ is the largest number $t$ such that for every $S\subseteq V$:
if $c(G\setminus S) > 1$, then $|S| \geq t \cdot c(G \setminus S)$.
\end{definition}

\noindent Intuitively, to have toughness $t$ means that you need to remove $t$ nodes for every component you hope to create. 
%
Win~\cite{win:1989} formalized this by establishing a  link between toughness and $k$-trees: 

\begin{theorem}[\cite{win:1989}]
For any $k\geq 3$, if $t(G) \geq \frac{1}{k-2}$, then $G$ has a $k$-tree.
\label{thm:win}
\end{theorem}

\indent Win's theorem captures the intuition that a small toughness indicates a small number of strategic node removals can generate a large
number of components. This in turn implies the existence of a spanning tree containing some high degree nodes (i.e., the nodes whose removal
creates many components).
We formalize this intuition 
with the following straightforward corollary of Win's theorem:

\begin{theorem}
Fix an undirected graph $G=(V,E)$ and degree $k \geq 3$. If $d(G) > k$, then there exists a non-empty subset of nodes $S \subset V$
such that $c(G \setminus S) > (k-2)\cdot |S|$.
\label{thm:toughness}
\end{theorem}

\begin{proof}
Since $d(G) > k$, the contrapositive of Thm.~\ref{thm:win} implies that 
$t(G) < 1/(k-2)$. 
By the definition of toughness,
there exists an $S \subset V$ such that $|S| = t(G) \cdot c(G \setminus S)$. 
For this set, $c(G\setminus S) = |S|/ t(G) > (k-2) |S|$.
\end{proof}
\fi

\section{Pairwise Capacity}
\label{sec:gk}

In their seminal paper~\cite{gupta:2000},
Gupta and Kumar approached the question of network capacity by considering 
the maximum throughput achievable for a collection of disjoint pairwise flows, each consisting of a single 
source and destination. They studied achievable capacity in both {\em arbitrary} networks as well
as {\em random} networks. In this section, we apply this approach to the mobile telephone model.

To do so, we formalize the {\em pairwise capacity} problem as the following constraint on the allowable flow sets (see Section~\ref{sec:problem}): for every pair $(s_i,R_i)\in F$, it must be the case that $R_i=\{x\}$ (i.e., $|R| =1$),
and neither $s$ nor $x$ shows up in any other pair in $F$.



\subsection{Arbitrary Networks}
\label{sec:mcf}



We begin by designing algorithms that (approximate) the maximum achievable throughput in an arbitrary network.  For now we will not focus on the convergence time, since our definition of capacity applies in the limit, so we describe the following as a centralized algorithm (the time required for each node to gather the full graph topology and run
this algorithm locally to generate an optimal routing schedule is smoothed out over time).  But as usual when considering centralized algorithms, we will care about the running time. 

Formally, we define the Pairwise Capacity problem to be the optimization problem where we are given a graph $G= (V, E)$ and a pairwise flow set $F$, and are asked to output a description of an (infinite) schedule which maximizes the throughput.  Our algorithm will in particular output a finite schedule which is infinitely repeated.  Our approach is to establish a strong connection between multi-commodity flow and optimal schedules, and then apply existing flow solutions as a step toward generating a near optimal solution for the current network. In other words, we give an approximation algorithm for Pairwise Capacity via a reduction to a multi-commodity flow problem. 

\begin{theorem}
There is a (centralized) algorithm for Pairwise Capacity that achieves throughput which is a $(3/2+\epsilon)$-approximation of the optimal throughput, for any $\epsilon > 0$. The convergence time is $n^{O(1)}{\epsilon^{-2}}$ and the running time is $n^{O(1)}{\epsilon^{-1}}$.
\label{thm:pairwise-appx}
\end{theorem}

\paragraph{Multi-Commodity Flow.}
In the \emph{maximum concurrent multi-commodity flow (MCMF)} problem, we are given a triple $(D,M,cap)$, where $D=(V_D,E_D)$ is a digraph, $M$ is collection $M \subseteq V_D \times V_D$ of node-pairs (each representing a \emph{commodity}),
and $cap:E_D \rightarrow \mathbb{R}^+_0$ are flow capacities on the edges.
Let $K = |M|$ be the number of commodities.
The output is a collection $f=(f_1, f_2, \ldots, f_K)$ of flows satisfying conservation and capacity constraints. Namely, for each flow $f_i$ and for each vertex $v \in G$ where $v \not\in \{s_i, t_i\}$, the flow into a node equals the flow going out: $\sum_{e=(u,v)\in E_D} f_i(e) = \sum_{e'=(v,w)\in E_D} f_i(e')$.
Also, the flow through each edge is upper bounded by its capacity: $f(e) = \sum_{i=1}^K f_i(e) \le cap(e)$.
Let $v(f_i) = \sum_{w, e=(s_i, w) \in E_D} f_i(e)$ be the \emph{value} of flow $i$, or the total flow of commodity $i$ leaving its source.  
The value of the total flow $f$ is $v(f) = \min_{i=1}^K v(f_i)$, and our goal is to maximize $v(f)$. 
We refer to $f$ as an \emph{MCMF flow} and the constituent commodity flows as \emph{subflows}.

The MCMF problem can be solved in polynomial-time by linear programming. There are also combinatorial approximation schemes known, and our version of the problem can be approximated within a $(1+\epsilon)$-factor in time $\tilde{O}((m+K) n/\epsilon^2)$ \cite{Madry2010}. 

We first show how to round an MCMF flow to use less precision while limiting the loss of value.
We say that a MCMF flow is \emph{$\phi$-rounded}
if the flow of each commodity on each edge is an integer multiple of $1/\phi$: $\lfloor f_i(e) \cdot \phi \rfloor = f_i(e)\cdot \phi$, for all $i$, and all edges $e$. We show how to produce a rounded flow of nearly the same value.

\begin{lemma}
Let $f$ be a MCMF flow and $\phi$ be a number. There is a rounding of $f$ to a $\phi$-rounded flow $f'$ with value at least $v(f') \ge v(f)(1 - K m/\phi)$, and it can be generated in polynomial time.
\label{lem:flow-rounding}
\end{lemma}
\begin{proof}
We focus on each subflow $f_i$.
By standard techniques, each subflow $f_i$ can be decomposed into a collection of paths $P_1, \ldots, P_s$ and values $\alpha_1, \ldots, \alpha_s$, with $s \le m = |E|$, such that 
$f_i(e) = \sum_{j, P_j \ni e} \alpha_j$ for each edge $e$. 
Let $\alpha'_j = \lfloor \alpha_j \cdot \phi\rfloor / \phi$, for each $j$, and 
observe that $\alpha'_j \ge \alpha_j - 1/\phi$. We form the $\phi$-rounded flow $f'$ by
$f'_i(e) = \sum_{j, P_j \ni e} \alpha'_j$, for each edge $e$. It is easily verified that conservation and capacity constraints are satisfied.
By the bound on $\alpha'$, it follows that the value of the rounded flow is bounded from below by
$v(f'_i)\ge v(f_i) - s/\phi \ge v(f_i) - m/\phi$. The value of each flow is trivially bounded from below by $v(f_i) \ge 1/K$ (which is achieved by sending $1/K$ of each commodity flow along a single path). Thus, $v(f'_i) \ge (1-Km/\phi)v(f_i)$.
\end{proof}

We now turn to the reduction of Pairwise Capacity to MCMF.
Given $G=(V,E)$ and $F$, along with a parameter $\tau$, we form the flow network $\bigD_\tau = (D, M, cap_\tau)$ as follows.
The undirected graph $G=(V,E)$ is turned into a digraph $D=(V_D,E_D)$ with two copies $v^{in}, v^{out}$ of each vertex: $V_D = \{v^{in}, v^{out} : v\in V\}$
and edges $E_D = \{ (u^{out}, v^{in}) : uv \in E\} \cup \{(v^{in},v^{out}) : v \in V\}$. 
The source/destination pairs carry over: $M = \{(s^{in}, t^{out}) : (s,\{t\}) \in F\}$.
Finally, capacities of edges in $E_D$ are
$cap_\tau(u_{out}, v_{in}) = \infty$ and $cap_\tau(v_{in}, v_{out}) = 1 + t_v\cdot \tau/2$, 
where $t_v$ is the number of source/destination pairs in $F$ in which $v$ occurs.
Observe that there is a one-to-one correspondence between simple paths in $G$ and in $D$ (modulo the in/out version of the start/end node).  


\begin{lemma}
The throughput of any schedule on $(G,F)$ is at most 
$\tau^*/2$, 
where $\tau^*$ is the largest value such that $\bigD_{\tau^*}$ has MCMF flow of value $\tau^*$. 
\label{lem:flow-lb}
\end{lemma}
\begin{proof}
Let $\bigA$ be a mobile telephone schedule and let $T$ be its throughput. 
We want to show that $\bigD_{2T}$ has MCMF flow of value $2T$; this is sufficient to imply the lemma.
We assume that packets flow along simple paths, and we achieve that by eliminating loops from paths, if necessary.
By the throughput definition, there is a round $r_0 = r^{\bigA,T}$ such that for every round $r \ge r_0$ and every source/destination pair $i$, the number of $i$-packets delivered by round $r$ is at least $T \cdot r$.
Let $X_i$ be the first $T r_0$ $i$-packets delivered (necessarily by round $r_0$), for each type $i$, and let $X = \cup_i X_i$.
For each edge $e=uv$ and pair $i$, let $q_i(u,v)$ be the number of packets in $X_i$ that passed through $e$, from $u$ to $v$. Also, for a vertex $v$, let $a_i(v)$ denote the number of $i$-packets originating at $v$, i.e., $a_i(v)=Tr_0$ if $v = s_i$ and $a_i(v)=0$ otherwise. Similarly, let $b_i(v)$ be the number of $i$-packets with $v$ as its destination.
Finally, let $q_i(v)$ be the number of packets in $X_i$ that flow through $v$, but did not originate or terminate at $v$, and observe that $q_i(v) = \sum_{w, vw\in E} q_i(v,w) - a_i(v) = \sum_{u, uv\in E} q_i(u,v) - b_i(v)$.

Define the collection $f=(f_1, f_2, \ldots, f_K)$ of functions where for each $i$, 
$f_i(u_{out}, v_{in}) = 2 q_i(u,v)/r_0$, for each edge $e=uv \in E$, and
$f_i(v_{in}, v_{out}) = 2(q_i(v) + a_i(v) + b_i(v))/r_0$, for each vertex in $V$.
Observe that the flow $f_i(u_{out}, v_{in})$ corresponds to twice the number of $i$-packets going from $u$ to $v$ (scaled by factor $1/r_0$). The flow $f_i(v_{in},v_{out})$ from $v_{in}$ to $v_{out}$ corresponds to the number of packets in $X_i$ coming into $v$ plus the number of those going out of $v$ (scaled by factor $1/r_0$), counting those that go through $v$ twice, but those originating or terminating at $v$ only \emph{once}.
We claim that $f$ is a valid MCMF flow in $\bigD_{2T}$ of value $2T$, which implies the lemma. Let $f_i^a(v) = 2a_i(v)/r_0$ ($f_i^b(v) = 2b_i(v)/r_0$) be the amount of type-$i$ flow originating (terminating) at $v$, respectively.

First, to verify flow conservation at nodes, consider a type $i$, and observe first that all packets in $X$ start at the source $s_i$ and end at the destination $t_i$. 
\begin{align*} 
f_i(v_{in}, v_{out}) &= \frac{2(q_i(v) + a_i(v) + b_i(v))}{r_0} = \frac{2a_i(v)}{r_0} + \sum_{u, uv\in E} \frac{2 q_i(u,v)}{r_0} \\
 &= f_i^a(v) + \sum_{u, (u_{out},v_{in})\in E_D} f_i(u_{out}, v_{in})   \ . \end{align*}
That is, the flow from each node $v_{in}$ equals the flow coming in plus the flow generated at the node (noting also that no flow terminates at the node). Similarly,
the flow into $v_{out}$ equals the flow terminating at the node plus the node going out:
\begin{align*} 
f_i(v_{in}, v_{out}) &= \frac{2(q_i(v) + a_i(v) + b_i(v))}{r_0} = \frac{2b_i(v)}{r_0} + \sum_{w, vw\in E} \frac{2 q_i(v,w)}{r_0} \\
 &= f_i^b(v) + \sum_{w, (v_{out}, w_{in})\in E_D} f_i(v_{out}, w_{in})   \ . \end{align*}

Second, to verify capacity constraints, observe that if $q(v) = \sum_i q_i(v)$ is the number of packets that flow through node $v$, then 
\[ 2q(v) + \sum_i (a_i(v) + b_i(v)) \le r_0\ , \]
since $v$ needs to handle flowing-through packets in two separate rounds and it can only process a single packet in a round. Thus, the flow through $(v_{in}, v_{out})$ is bounded by
\[ f(v_{in}, v_{out}) 
  = \frac{2}{r_0} (q(v) + \sum_i (a_i(v)+b_i(v))) 
  = \frac{1}{r_0} (2q(v) + \sum_i (a_i(v)+b_i(v))) + t_v T 
  \le 1 + t_v T\ ,\]
satisfying the capacity constraints.

Finally, it follows directly from the definition of $f_i^a$ (or $f_i^b$) that the flow value is $2T$.
\end{proof}

To prove Theorem~\ref{thm:pairwise-appx} we need to introduce edge multicoloring. 

\begin{definition}
Given a graph $G=(V,E)$ and a \emph{color requirement} 
$r(e) \in \mathbb{N}$ for each edge $e \in E$. 
An \emph{edge multicoloring} of $(G,r)$
is a function $\pi: E \rightarrow 2^\mathbb{N}$ that satisfies the following: 
a) if $e_1,e_2\in E$ are adjacent then $\pi(e_1) \cap \pi(e_2) = \emptyset$, and
b) $|\pi(e)| \ge r(e)$, for each edge $e \in E$.
The number of colors used is $|\cup_e \pi(e)|$, the size of the support for $\pi$.
\end{definition}

We shall use the follow result on edge multicolorings.

\begin{theorem}[Shannon \cite{Shannon1949}]
Given a graph $G=(V,E)$ and a \emph{color requirement} $r(e) \in \mathbb{N}$ for each edge $e \in E$, 
there is a polynomial-time algorithm that edge multicolors $(G,r)$ using at most $3\Delta_r(v)/2$ colors,
where $\Delta_r(v) = \sum_{e \ni v} r(e)$.
\label{thm:shannon-edgemulticol}
\end{theorem}

We can now prove Theorem~\ref{thm:pairwise-appx}.

\begin{proof}[Proof of Theorem~\ref{thm:pairwise-appx}]
Let $(G,F)$ be a given Pairwise Capacity instance and let $\epsilon > 0$. 
We perform binary search to find a value $\tau$ such that:
a) An $1+\epsilon/4$-approximate MCMF algorithm produces flow $f$ of value at least $\tau(1-\epsilon/4)$ on $\bigD_\tau$, and
b) The same does not hold for $\tau (1+\epsilon/4)$. 
The resulting flow $f = (f_1, \ldots, f_K)$ is then of value at least $\tau^*(1-\epsilon/4)/(1+\epsilon/4) \ge \tau^*(1-\epsilon/4)^2 \ge \tau^*(1-2\epsilon/4)$. Recall that $K$ is the number of commodities, and so $K = |F|$. 

Let $N = 4 \epsilon^{-1} K m$. We apply Lemma \ref{lem:flow-rounding} to create from $f$ an $N$-rounded flow  $f'=(f'_1, \ldots, f'_K)$. 
By Lemma \ref{lem:flow-rounding}, 
this decreases the flow value by a factor of at most $1-K m/N = 1-\epsilon/4$, i.e., $v(f') \ge (1-\epsilon/4)v(f) \ge (1-3\epsilon/4)\tau^*$. 

We then form an edge multicoloring instance on $G$ as follows. Each edge $e$ requires 
$r(e)$ colors, where $r(e) = \sum_i r_i(e)$ and $r_i(e) = f'_i(e)\cdot N$. 
The weighted degree of each node $v$ is then $d_r(v) = \sum_{e \ni v} r(e) 
= N \sum_{e \ni v} f'(e) \le N \sum_{e \ni v} f(e) = \sum_{e=(v,u)\in E_D} f(e) + \sum_{e'=(w,v) \in E_D} f(e) \le 2N$, by node capacity constraints. 
We apply the algorithmic version of Shannon's Theorem \ref{thm:shannon-edgemulticol} to edge multicolor $(G,r)$ with at most $3N$ colors.
This induces an initial schedule of length $3N$, which is then repeated as needed.
Within each $3N$ rounds, $\sum_{v, e=(s_i,v)} r_i(e) = N \cdot v(f'_i)$ $i$-packets depart from its source $s_i$. 

Let $r_0 = \frac{4n}{\epsilon} (3N)$. 
Consider the situation after round $r \ge r_0$.  
Observe that each packet is forwarded at least once during each $3N$ rounds, and thus it is delivered within $n (3N)$ rounds after it is transmitted from its source, since each path used is simple.
Thus, the total number of type-$i$ packets that remain in the system in the end is at most a $\epsilon/4$-fraction of the delivered packets.  
Averaged over the $r$ rounds gives throughput of
\[ \frac{N \cdot v(f'_i)}{3N} \cdot (1-\epsilon/4) = 
\frac{1}{3}v(f'_i) \cdot (1-\epsilon/4) \ . \]

Hence, the throughput achieved is at least
\begin{equation}
    T \ge \frac{1}{3} v(f') (1-\epsilon/4)\ge \frac{1}{3} v(f) (1-\epsilon/4) \ge \frac{1}{3} \tau^* (1-\epsilon)\ .
\label{eq:flowcap}
\end{equation}
By Lemma \ref{lem:flow-lb}, the throughput is then $3/2 + \epsilon$-approximation of optimal. 

The computation performed is dominated by the application of Shannon's algorithm, which runs in time $O((\Delta_r+n) \hat{m})$, where $\hat{m}$ is the number of multiedges and $\Delta_r \le 2N$ is the maximum weighted degree. Here, $\hat{m} = \sum_e q(e) = N \sum_e \sum_i f_i(e) \le N \cdot m$.
Hence, the number of computational steps is at most $O(m N^2) = O(m^3 K^2 \epsilon^{-2})$.
The convergence time is $r_0 = \frac{4n}{\epsilon} (3N) = O(n m K \epsilon^{-2})$.
\end{proof}
We note that the factor $3/2$ cannot be avoided in a reduction to flow. Consider the graph $G$ on six vertices $V=\{s_i, t_i: i=0,1,2\}$ and edges $\{s_i t_{i'}, t_i t_{i'} : i=0,1,2, i' = i-1 \bmod{3}\}$. The optimal throughput is $1/3$, with respect to $F = \{(s_i,t_i):i=0,1,2\}$. 
This corresponds to the directed graph $D$ on nine nodes: $\{s_i, t_i^{in}, t_i^{out} : i=0,1,2\}$ and edges $\{(s_i, t_{i'}^{in}), (t_{i'}^{out}, t_i^{in}), (t_i^{in}, t_i^{out}) : i=0,1,2, i'= i-1 \bmod{3}\}$, and three subflows: $M=\{s_i, t_i^{out} : i=0,1,2\}$. Then, $\bigD_1 = (D,M, cap_1)$, where $cap_1(t_i^{in}, t_i^{out}) = 2$, has flow of value 1. 

\subsection{Random Networks}
\label{sec:gk:random}
We now consider achievable throughput for the pairwise capacity problem in networks
randomly generated with the $GK$ process defined in Section~\ref{sec:problem}.
Following the lead of the original Gupta and Kumar capacity paper~\cite{gupta:2000},
we assume the flow sets are also randomly generated with uniform randomness
and contain all the nodes (i.e., every node shows up as a source or destination).
A minor technical consequence of this definition
is that it requires us to constrain our attention to even network sizes.

We begin in Section~\ref{sec:gk:conn} by identifying a threshold value for
the radius $r$ below which the randomly generated network
is likely to be disconnected, trivializing the achievable throughput to $0$.
In Sections~\ref{sec:gk:upper} and~\ref{sec:gk:lower},
we then prove that for any radius value $r$ that is at least a sufficiently
large constant factor greater than the threshold,
with high probability in $n$, 
the optimal achievable throughput is in $\Theta(r)$.

\subsubsection{Connectivity Threshold}
\label{sec:gk:conn}

When analyzing networks and flows generated by the $GK(n,r)$ network generation process,
we must consider the radius parameter $r$. If $r$ is too small, then we expect a network in which some sources
are disconnected from their corresponding destinations, making the best achievable throughput trivially $0$.
Here we study a {\em connectivity threshold} value $r_c(n) = \sqrt{\frac{\alpha\log{n}}{n}}$,
defined with respect to a network size $n$ and a constant fraction  $\alpha$.
We prove that for any $r \leq r_c(n)$,
with probability at least $1/2$,
given a network generated by $G(n,k)$ and a random pairwise flow set $F$,
there exists at least one pair in $F$ that is disconnected. 

\begin{theorem}
There is some constant $\alpha > 0$ so that for every sufficiently large even network size $n$
and radius $r \leq r_{c}(n) =\sqrt{\frac{\alpha\log{n}}{n}}$, if $G \sim GK(n, r)$ and $F$ is a random pairwise flow set, then with probability at least $1/2$ there exists $(s,\{x\}) \in F$ such that $s$ is disconnected from $x$ in $G$. 
\label{thm:gk:threshold}
\end{theorem}

At a high level, 
to prove this theorem
we divide the unit square into a grid consisting of {\em boxes} of side length $r$,
and then group these boxes into {\em regions} made up of
$3\times 3$ collections of boxes. 
If a given region has a node $u$ in the center box, and all its other boxes are empty, then $u$ is disconnected from any node not 
in its own box.
Our proof calculates that for a sufficiently small constant fraction $\alpha$ used in the definition
of the connectivity threshold,
with probability at least $1/2$,
there will be a node $u$ such that $u$ is isolated as described above,
and $u$ is part of a source/destination pair with another node $v$ located in a different box.

Given this setup, the main technical complexity in the proof is carefully navigating the various
probabilistic dependencies.
One place where this occurs is in proving the likelihood of empty regions.
For sufficiently small $\alpha$ values,
the expected number of non-empty regions is non-zero, but we cannot directly concentrate on this expectation due to the dependencies between emptiness events. 
These dependencies, however, are  dispatched by  leveraging the 
negative association
between the indicator variables describing a region's emptiness (e.g., if region $i$ is not empty,
this {\em increases} the chance that region $j\neq i$ is empty).
\iflong
In particular, we will apply the following results concerning 
negative association derived in~\cite{daum:2012} based
on the more general results of~\cite{dubhashi:1998}:

\begin{theorem}[\cite{daum:2012,dubhashi:1998}]
Consider an experiment in which weighted balls are thrown into $n$ bins according to some distribution.
Fix some $a \geq 0$,
and let $Y_i$ be the indicator random variable defined such that $Y_i = 1$ iff there are no
more than $a$ balls in bin $i$. 
The variables $Y_1, Y_2, ..., Y_n$ are negatively associated, and
therefore standard Chernoff bounds apply to their sum.
\label{thm:gk:negative}
\end{theorem}

We can proceed to the main proof:

\begin{proof}[Proof (of Theorem~\ref{thm:gk:threshold}).]
We consider the network generated with the threshold connectivity value $r= r_{c}(n) =\sqrt{\frac{\alpha\log{n}}{n}}$
defined in the theorem statement. Clearly if the network is disconnected for this
radius it is also disconnected for smaller radii.
We will show the theorem claim holds for $r_c(n)$ for sufficiently large $n$ and $\alpha = 1/32$.

We begin by structuring the unit square into which nodes are randomly placed by the $GK$ process.
First, we divide the unit square into a grid of square {\em boxes} of side length $r$ (ignore left over space).
We then partition these {boxes} into {\em regions} made up of $3\times 3$ collections of boxes
(ignore left over boxes).
Finally, we label these regions $1,2,...,k$, where
\begin{align*}
	k & = \lfloor \text{(\# of boxes)}/9  \rfloor =  \lfloor  ((1/r)^2)/9 \rfloor =  \left \lfloor  \left(\sqrt{\frac{n}{\alpha\log{n}}}\right)^2/9 \right\rfloor = \lfloor  n/(9\alpha\log{n})\rfloor.
\end{align*}

For each region $i$, 
let $c_i$ refer to the center box of the $3\times 3$ pattern of boxes that defines the region. 
We call the remaining $8$ boxes the {\em boundary}
boxes for region $i$.  
We now calculate the probability that $GK(n, r_c(n))$ process places
nodes such that boundary boxes of a given region $i$ are all empty.

By the definition of the $GK(n,r_c(n))$ process,
the probability that a given node $u$ is placed in a given box is equal to the total area, $a_b$,
of the box.
Therefore, the probability $u$ is not placed in {\em any} of the $8$ boundary boxes of a given
region is $1-8\cdot a_b$.

Pulling these pieces together with the fact that  $a_b = r_c(n)^2 = (\alpha\log{n})/n$,
it follows that the probability that {\em no} node is placed in the boundary
boxes of a given region $i$ is lower bounded as:

\begin{align*}
\Pi_{u\in V} (1-8a_b) &= (1-8a_b)^n \geq (1/4)^{8\cdot a_b \cdot n} = (1/4)^{8\alpha\log{n}},
\end{align*}
where the second step follows from the well-known inequality that $(1-p) \geq (1/4)^p$ for any $p \leq 1/2$ (for sufficiently large $n$, it is clear that $p=8a_b \leq 1/2$).  Because we assumed $\alpha = 1/32$, we can further simplify:
\begin{align*}
	(1/4)^{8\alpha\log{n}} &= (1/4)^{(1/4)\log{n}} = \left((1/4)^{1/2}\right)^{(1/2)\log{n}}  = (1/2)^{\log{n^{1/2}}} =  1/\sqrt{n}
\end{align*}
	
 We now lower bound the probability that \emph{some} region has empty boundary boxes.  To do this, we first define the random indicator variables $Y_1, Y_2, ..., Y_k$, where $Y_i = 1$ iff the boundary boxes of region $i$ are empty.	Let $Z=\sum Y_i$. We want to lower bound the probability that $Z>0$. By linearity of expectations, 
\[
\E(Z) \geq k/\sqrt{n} = \left\lfloor \frac{\sqrt{n}}{9\alpha\log{n}} \right\rfloor.
\]
	
Because each node is equally likely to be placed in each region, we know from Theorem~\ref{thm:gk:negative} (with $a=0$) that the variables $Y_1, Y_2, ..., Y_k$, are negatively associated.  Therefore the Chernoff bounds from Theorem~\ref{thm:chernoff} apply to $Z$.
	In particular, it follows that the probability that $Z \leq \E(Z)/2$ is upper bounded
	by:
	
	\[ \exp{\left(- \E(Z)/8 \right)} \leq \exp{\left( -  \frac{\sqrt{n}}{72\alpha\log{n}}  \right)}\]
	
	\noindent For our fixed $\alpha = 1/32$,
	it follows that 
	for sufficiently large $n$, two things are true:
	$\E(Z) \geq 2$ (and therefore $\E(Z)/2 > 0$), and this probability is upper bounded by $1/4$.
	Therefore, for sufficiently large $n$, the probability that there are no regions
	with empty boundary boxes is at most $1/4$.

	
	Conditioned on the event that a given region $i$ has empty boundary boxes,
	we want to now bound the probability that there exists a source/destination pair $(u, \{v\})\in F$ such that $u$
	is in $c_i$ and $v$ is not.
	
	For a given $(u,\{v\}) \in F$, this occurs with probability $p_1p_2$,
	where $p_1$ is the probability that $u$ is in $c_i$ and $p_2$ is the probability
	that $v$ is not in $c_i$. Given that $p_1 = a_b$ (where $a_b$ is the area contained
	in a box) and $p_2$ is clearly greater than $1/2$, we crudely bound this product as
	\[ p_1p_2 > a_b/2 = \frac{\alpha \log{n}}{2n}.\]
	So the probability that this splitting event {\em fails} to occur for all $n/2$ pairs in $F$ is upper bounded by
    \begin{align*} 
    (1-p_1p_2)^{n/2} &< \left(1-\frac{\alpha \log{n}}{2n}\right)^{n/2} \leq e^{-(1/4)\alpha\log{n}},
    \end{align*}
    As before, for our fixed $\alpha = 1/32$, for sufficiently large $n$ this
    probability is upper bounded by $1/4$.

    We have shown the following two bounds: (1) the probability that there are no regions with empty boundary
    boxes is at most $1/4$; and (2) the probability that given a region with empty boundary boxes, that there
    are no pairs split by the region, is also at most $1/4$. We can combine these events with a union
    bound to establish that the probability
    that at least one of these two events fails is less than $1/2$, satisfying the theorem statement.
\end{proof}
\fi

\subsubsection{Bound on Achievable Throughput}
\label{sec:gk:upper}

In the previous section, we identified a radius threshold $r_c(n)$ below which a randomly generated network
is likely to disconnect a source and destination, reducing the achievable throughput to a trivial $0$.
Here we study the properties of the networks generated with radius values on the other side of this threshold.
In particular, we show that for any radius $r \geq r_c(n)$, 
with high probability, the randomly generated network and flow set will allow an optimal
throughput bounded by $O(r)$. 
The intuition for this argument is that if nodes are evenly distributed
in the unit square, a constant fraction of senders will have to deliver packets
from one half of the square to the other, necessarily requiring
many packets to flow through a small column in the center of the square, bounding the achievable throughput.


\begin{theorem}
For every sufficiently large even network size $n$ and radius $r \geq r_c(n)$, given a network $G \sim GK(n,r)$ and a random pairwise flow set $F$, 
the throughput of every schedule (w.r.t.\ $G$ and $F$) is $O(r)$ with high probability. 
\label{thm:gk:optimal}
\end{theorem}
\iflong
To build up to this proof,
we consider a series of helper lemmas.
These results assume that we divide the unit square into three columns (regions of height 1) such that the center region has width $r$ and the two outer regions width $(1-r)/2$. We first show that, in expectation, there are many source/destination pairs such that all paths between the source and destination require a node in the center region to send a packet to a node in an outer region. Slightly more formally, we say that a source/destination pair $(s_i, t_i)$ requires a node in the center region if every path from $s_i$ to $t_i$ in $G$ contains at least one node from the center region.  

For the lemmas that follow, since the theorem is trivially true for constant $r$, we assume without loss generality, that $r$ is relatively small (e.g., $r<1/2$).

\begin{lemma}  \label{lem:leavesprobability}
For a particular source/destination pair $(s_i, t_i)$, the probability that $(s_i, t_i)$ requires a node from the center region is at least $\frac{1}{2}(1-r^2)$.
\end{lemma}
\begin{proof}
Note that $(s_i, t_i)$ requires a node in the center region if one of the following two disjoint events occur: $s_i$ and $t_i$ are in different outer regions, or $s_i$ is in the center region but $t_i$ is in an outer region.  The first event is sufficient since the width of the center region means that there are no edges between the two outer regions, while the second event is sufficient since every $s_i - t_i$ path includes $s_i$.  

The first event occurs with probability $2(((1-r)/2)^2 = (1-r)^2 / 2$, and the second event occurs with probability $r(1-r)$.  Thus the total probability that every $s_i - t_i$ path includes an outgoing edge from a node in the center region is at least $(1-r)^2 / 2 + r(1-r) = \frac12 (1-r^2)$. 
\end{proof}

Next we relate this probability to the number of such source/destination pairs.  
\begin{lemma} \label{lem:pairchernoff}
With very high probability, the number of source/destination pairs in $F$ that require a node in the center region 
is at least $\Omega (n(1-r))$. 
\end{lemma}
\begin{proof}
For each source/destination pair $(s_i, t_i)$, let $X_i$ be an indicator random variable for the union of the two events analyzed in Lemma~\ref{lem:leavesprobability}, such that $\E[X_i] \geq \frac12 (1-r^2)$.  Observe that, clearly, these events are independent
and let $X = \sum_{i=1}^{n/2} X_i$ denote the total number of pairs where $X_i=1$.  By linearity of expectations, we know that $\E[X] \geq (n/4)(1-r^2)$.  So the Chernoff bound from Theorem~\ref{thm:chernoff} implies that 
\begin{align*}
    \Pr[X < (n/8)(1-r^2)]\leq \Pr[X < \mu/2]\leq \exp{\left(\frac{-n(1-r^2)}{16}\right)}.
\end{align*}
Therefore, with very high probability, the number of source/destination pairs that meet the conditions of Lemma~\ref{lem:leavesprobability} is $\Omega((n/8)(1-r^2))=\Omega(n(1-r^2))=\Omega(n(1-r))$ for $r < 1/2$. Furthermore, since by Lemma~\ref{lem:leavesprobability} each of these source/destination pairs requires a node in the center region, the number of pairs as described by the lemma statement is also $\Omega(n(1-r))$.
\end{proof}

Now that we have successfully lower bounded the number of source/destination pairs that require a node in the center region to send a packet to a node in an outer region, we need an estimate for how many nodes in the center region exist to send these packets at one time.

\begin{lemma}\label{lem:centernodes}
With high probability, there are $O(rn)$ nodes in the center region.
\end{lemma}
\begin{proof}
Let $Y$ be a random variable denoting the number of vertices in the center region.  Each node is put into the center region independently with probability $r$, and thus $\E[Y] = rn$.  Since the placement of each node is independent, we can use the Chernoff bound from Theorem~\ref{thm:chernoff} to get that $\Pr[Y \geq 2rn] \leq \exp{(-rn/3)}$.  Thus with very high probability, there are at most $2rn = O(rn)$ nodes in the center region.
\end{proof}


We now have everything we need to upper bound the pairwise throughput.

\begin{proof}[Proof (of Theorem~\ref{thm:gk:optimal}).]
From Lemma $\ref{lem:pairchernoff}$ we know that with high probability that there are $\Omega(n(1-r))$ source/destination pairs that require one of the $O(rn)$ nodes in the center region. Since each of these nodes can send at most one packet per round by the constraints of the mobile telephone model, by round $\Psi$ at most $O(\Psi\cdot rn)$ packets can be delivered. Therefore, on average for each source/destination  pair $(s_i, t_i)$, the number of packets delivered by round $\Psi$ is $O(\Psi\cdot rn)/\Omega(n(1-r))=O(\Psi r)$.  Thus in any schedule there must exist some $(s_i, t_i)$ so that at round $\Psi$, only $O(\Psi r)$ packets from $s_i$ have been delivered to $t_i$, and hence the throughput is only $O(r)$.
\end{proof}


\fi

\subsubsection{Tightness of the Throughput Bound}
\label{sec:gk:lower}

In Section~\ref{sec:gk:upper},
we proved an upper bound of $O(r)$ on the achievable throughput
in a network generated by $GK(n,r)$, for $r \geq r_c(n)$, and random pairwise
flows. Here we show
this result is tight by showing how to produce a schedule that achieves throughput in $\Omega(r)$ with respect to 
a random $G$ and $F$.
Formally:
\begin{theorem}
There exists a constant $\beta > 1$ such that, for any sufficiently large network size $n \geq 2$ 
and radius $r \geq \beta r_c(n)$, if $G \sim GK(n,r)$ and $F$ is a random pairwise flow set, then with high probability in $n$ there exists a schedule that achieves throughput in $\Omega(r)$ with respect to $G$ and $F$.
\label{thm:gk:existential}
\end{theorem}

At a high level, 
our argument divides the unit square into box of side length $\approx r$.
We prove that with high probability, both nodes and pairwise demands are evenly distributed among the boxes.
This allows a schedule that efficiently moves many packets in parallel up and down columns to the row of their destination,
and then moves these packets left and right along the rows to reach their destination.
The time required for a given packet to make it to its destination is bounded by the column and row length
of $\approx 1/r$, yielding an average throughput in $\Theta(r)$.
The core technical complexity of this argument is the careful manner in which packets are moved onto
and off a set of parallel paths while avoiding more than a small amount of congestion at any point in their routing.

\iflong
Our approach is to isolate the probabilistic elements of the proof. To do so, 
we need some preliminary definitions to help structure our argument.
We begin by fixing a canonical way of covering the unit square into which
the $GK$ process places nodes with a grid.

\begin{definition}
Fix some radius $r, 0 < r \leq 1$.
An {\em $r$-grid} is a partition of the unit square into {\em boxes} of
side length $\hat r$, where $\hat r$ is the largest value such that: 
(a) $\hat r \leq r$; (b) $k=1/\hat r$ is a whole number; and (c) the distance between any
two locations in boxes that share a side is at most $r$ (i.e., $\hat r \leq r/\sqrt{5}$).
We call $\hat r$ the {\em grid radius} and $k$ the {\em grid size} of the $r$-grid.
\end{definition}

We next define some useful properties of node placements and flow set definitions
with respect to this grid structure.

\begin{definition}
Fix some even network size $n \geq 2$, and radius $r>0$.
Let $G \sim GK(n,r)$, and let $F$ be a random pairwise flow set.
Consider the $r$-grid with grid size $k$. For each $i,j \in [1,k]$,
we define the following two random variables:

\begin{itemize}
    \item $X_{i,j}$ is the number of nodes placed in the grid box in row $i$ and column $j$.
    \item $Y_{i,j}$ is the number of pairs $(s_{\ell}, \{t_{\ell}\}) \in F$ where $s_{\ell}$ is placed in column $j$ and $t_{\ell}$ is placed in row $i$.
\end{itemize}

\noindent We say $G$ and $F$ are {\em good}
if $(1/4)\cdot n/k^2 \leq X_{i,j}, Y_{i,j} \leq 4\cdot n/k^2$.
\label{def:good}
\end{definition}

Now we show that $G$ and $F$ are good with high probability as long as the radius is a sufficiently large
constant factor larger than the connectivity threshold $r_c(n)$ (see Section~\ref{sec:gk:conn}).

\begin{lemma}
There exists a constant $\beta > 1$ such that, for any sufficiently large network size $n \geq 2$ 
and radius $r \geq \beta r_c(n)$, if $G \sim GK(n,r)$ and $F$ is a random pairwise flow set, then $G$ and $F$ are good with high probability in $n$.
\label{lem:gk:good}
\end{lemma}
\begin{proof}
Suppose that we first pair up the nodes (each pair consisting of a source and a destination) and then randomly place the nodes in the unit square.  Clearly this is equivalent to first running $GK$ and then choosing a random pairwise flow set, so we analyze this modified process here.  
To do so, first fix an $n$ and $r \geq \beta r_c(n)$ as specified in the lemma statement,
where we will bound the specific constant $\beta$ we need later in this proof.
Consider the $r$-grid with grid radius $\hat r$ and grid size $k$.
Finally, fix a row $i$ and column $j$ from this grid.

We begin by bounding $Y_{i,j}$.
Let $F$ be our predefined set of pairwise flows (i.e., pairs of sources and destinations).
Label these pairs arbitrarily, $1,2,...,(n/2)$.
We define $Y_{i,j} = \sum_{\ell = 1}^{n/2} Z_{\ell}$, where $Z_{\ell}$ is the random indicator variable that equals
$1$ iff the source for pair $\ell$ in $F$ is placed in grid column $j$ and the destination for pair $\ell$ is placed
in grid row $i$.

Each grid row and column takes up a $1/k$ fraction of the union square.
It follows that $\Pr(Z_{\ell} = 1) = (1/k)^2$, and therefore, by linearity of expectation: 
$\E(Y_{i,j}) = \sum_{\ell=1}^{n/2} \E(Z_{\ell}) = n/(2k^2)$.
Because $Y_{i,j}$ is the sum of independent indicator random variables, 
we can apply our Chernoff forms from Theorem~\ref{thm:chernoff} to concentrate on this expectation.
In particular, if we use parameters $\mu= \E(Y_{i,j}) = n/(2k^2)$ and $\delta =1/2$, 
and also apply the loose bound $\hat r \geq r/6$, we get the following
from the lower bound form of Theorem~\ref{thm:chernoff}:

\begin{eqnarray*}
\Pr(Y_{i,j} \leq (1/2)\mu) & \leq & \exp{\left(\frac{- \mu}{8}\right)}\\
& = & \exp{\left(\frac{- n}{16k^2}\right)} \\
& = &\exp{\left(\frac{- n \hat r^2}{16}\right)} \\
& \leq & \exp{\left(\frac{- n (\beta r_c(n))^2}{96}\right)}  \\
& = & \exp{\left(\frac{- n \left(\beta \sqrt{\frac{\alpha\log{n}}{n}}\right)^2}{96}\right)}  \\
& = & \exp{\left(\frac{- n \beta^2 \frac{\alpha\log{n}}{n}}{96}\right)}  \\
& = & \exp{\left(\frac{- \beta^2 \alpha\log{n}}{96}\right)}  \\
\end{eqnarray*}

Similarly, if we instead set $\delta =1$, we get the following from the upper bound form:

\[\Pr(Y_{i,j} \geq 2\mu) \leq \exp{\left(\frac{-\mu}{3}\right)} \leq \exp{\left(\frac{- \beta^2 \alpha\log{n}}{36}\right)},\]

\noindent where the last step follows from adapting the lower bound derivation above to replace the $8$ with a $3$ in
the initial step.
For any constant $c\geq 1$,
there is a sufficiently large constant value for $\beta$,
defined independently of $n$,
such that both these probabilities are less than $n^{-c}$ (e.g., $\beta \geq \sqrt{(96c)/\alpha}$).
Call this value $\beta(c)$.

We now turn our attention to $X_{i,j}$. 
Our process for generating $G=(V,E)$ and $F$ places each node with uniform randomness in the unit square.
With this in mind, for each $u\in V$, let $Z'_u$ be the independent random indicator variable that equals $1$ iff node $u$ is
placed in grid box $(i,j)$. We can then define $X_{i,j} = \sum_{u\in V} Z'_u$.
Because each grid box has area $1/k^2$, it follows that $\Pr(Z'_u = 1) = 1/k^2$, and therefore,
by linearity of expectation, $\E(X_{i,j}) = n/k^2$.
Because $\E(X_{i,j}) = 2\E(Y_{i.j})$, and $X_{i,j}$ is also the sum of independent random indicators,
the same concentration bounds derived above for $Y_{i,j}$ still apply to $X_{i,j}$,
albeit now more loosely than before (the slightly larger expectation intensifies the concentration).

To conclude the proof, assume our goal is to end up with a failure probability less than $n^{-c}$, for some
constant $c\geq 1$.
We show that $\beta(c+3)$ is a sufficiently large definition of constant $\beta$ to satisfy the lemma
statement for this bound.

To do so, we first note that Definition~\ref{def:good} requires that every $X_{i,j}$ and $Y_{i,j}$
be within the range $[\mu/4, 4\mu]$, for $\mu = n/k^2$. Applying our above bounds
with constant $\beta(c+3)$, it follows that any given $X_{i,j}$ or $Y_{i,j}$ is within
this range with probability at least $1-n^{-(c+3)}$. 
By a union bound, the probability this fails to hold for {\em any} $X$ is less than $1/n^{c+2}$,
and the same holds for {\em any} $Y$.
A final union bound provides that the probability either the $X$ or $Y$ condition fails is
itself still less than $1/n^c$, as required.
\end{proof}

Our final result of this section proves that if $G$ and $F$ are good,
then there exists a schedule that achieves throughput in $\Omega(r)$.
To do, we describe an algorithmic process for generating this schedule.
This result is existential because the process makes use of the locations
in the unit square used by $GK$ to generate $G$. 

\begin{lemma}
Let $G \sim GK(n,r)$ for some even network size $n\geq 2$ and radius $r>0$,
and let $F$ be a random pairwise flow set.
If $G$ and $F$ are good then there exists a schedule that achieves throughput in $\Omega(r)$
with respect to $G$ and $F$.
\label{lem:gk:routing}
\end{lemma}
\begin{proof}
Fix some $n$, $r$, $G=(V,E)$, and $F$ as specified by the lemma statement.
Assume $G$ and $F$ are good.
We now construct a schedule for $G$ and $F$ that makes use of the
location that the $GK(n,r)$ process placed each node in the unit square.

In more detail, cover the unit square with an $r$-grid of grid radius $\hat r$ and grid size $k$.
By the definition of good, 
each grid box $(i,j)$ (i.e., the box in row $i$ and column $j$ of the $r$-grid)
contains at least $n_{low} = (1/4)n/k^2$ nodes and no more than $n_{high} = 4n/k^2$ nodes.
In each such $(i,j)$, choose $n_{low}$ nodes to be the {\em core}
nodes for this box. Label them arbitrarily, $1,2,...,n_{low}$.
Routing between adjacent boxes in our strategy will always take place on links
between core nodes in these boxes with the same label.

We now describe a multi-phase process that routes one packet from each source to each destination
in $O(k)$ rounds.
This process can then be repeated for each new packet, waiting for the current packet to be completely
delivered for all nodes before moving on to the next.
This will result in an overall throughput in $\Omega(1/k) = \Omega(r)$.
To simplify discussion, we will use the following notions of directionality:
{\em up}  means moving from larger to smaller row numbers,
{\em down} means smaller to larger rows,
{\em left} means larger to smaller columns, and {\em right} means smaller to larger columns.

The first phase is dedicated to routing packets up their current column in the grid 
to arrive at the {\em destination} row containing the packet's destination.
This phase only applies to packets that start below  their destination row (i.e., in a larger number row).
To do so, we first put aside $15$ rounds for non-core source nodes to send their packet to a core
node in their grid box. We assign non-core nodes to core nodes evenly.
Because there are $n_{low}$ core nodes, and at most $16n_{low}$ total nodes (leaving at most $15n_{low}$
non-core sources),
this load balancing ensures that $15$ rounds are sufficient,
and no core node has been assigned more then $16$ total packets (including its own).

We now route packets up the grid columns. To do so, the core node with label $\ell$ in grid
box $(i,j)$ (for $i>1$) forms a link with the core node with label $\ell$ in $((i-1), j)$.
This forms a {pipeline} of nodes of the same label in each column.\footnote{We assume
here the variation of the mobile telephone model in which you can have one outgoing and one incoming
link per round. If we instead assume the slightly more restrictive version where there is at most
one total link then we can easily simulate the former model at the cost of a factor of $2$ more rounds.}
Notice that $16$ rounds is enough for the core nodes in a given grid box,
to route all of their packets up to their corresponding core nodes in the grid box above.
Therefore, $16k$ rounds is enough to ensure every packet that needs to move up to get to
its destination row has successfully arrived at its destination row.

There is, however, a subtle complication that must be addressed.
Assume we are unlucky and that many (say, a super-constant number of) packets in a given column have destination row $i$,
and that they all happen to be assigned to core nodes with the same label $\ell$.
When the {\em up routing} phase as described above completes, these packets will have all successfully arrived
at row $i$, {\em but} they will only be known by core node $\ell$ in this box. Before we can successfully
route them left or right in subsequent phases, we would then have to spend too much time spreading
them out from core node $\ell$.

To avoid this complication, we add a {\em balancing} step to the up routing.
As stated above, our routing takes place in groups of $16$ rounds,
where in each such group, each box moves its current packets up to the box above.
We now introduce a rebalancing procedure in between each of these groups.
In more detail, fix any grid box $(i,j)$ and core node $\ell$.
If node $\ell$ has received {\em more} than $16$ packets with destination row $i$, 
then $\ell$ will evenly distributed these extra packets among other core nodes in its
grid box, one by one, until its count is back down to $16$.

We know this rebalancing is always possible as the definition of good
provides that $Y_{i,j} \leq 16n_{low}$,
so there is always room to rebalance packets to keep each of the $n_{low}$ 
core nodes count at $16$ or below.
Because at most $16$ new packets can arrive at a given core node in
each group of routing rounds, an additional
$16$ rounds is always sufficient to complete this rebalancing.

Combined, it follows that $O(k)$ rounds are sufficient not only to complete the up routing,
but to also ensure that packets are evenly distributed among core nodes at their destination row.
We follow this up routing phase with a symmetric down routing phase,
that routes packets that start above their destination row down to their destination row.
This requires an additional $O(k)$ rounds.
When these two phases are done, each packet is assigned to a core node
in its destination row, and no core node is assigned more than a constant
number of packets.

To complete the routing, 
we now turn our attention to moving packets across columns.
We being by using the above procedure to move packets to the left.
That is, this phase applies to packets that are in a column to the right of
their destination column.

For this phase, we replace the rebalancing steps with {\em delivery} steps.
In more detail, when a packet $p$ arrives at a core node in the grid box containing
its destination, the core node will deliver it to its destination during the next delivery step.
Because each core node can receive at most $16$ new packets per routing group, 
$16$ rounds is sufficient for the delivery step.

To conclude the routing, after routing packets right to left,
we execute a final phase that moves packet left to right.
The total time required to complete this routing of a single packet over all pairwise
flows is the time required by the four routing phases.
Because each phase requires $\Theta(k)$ rounds, the total time is $\Theta(k)$.
\end{proof}

To conclude
this section,
we note that the correctness of Theorem~\ref{thm:gk:optimal} is a direct corollary of Lemmas~\ref{lem:gk:good} and~\ref{lem:gk:routing}.

\fi



\section{Broadcast Capacity}
\label{sec:bcast}

The broadcast capacity problem assumes a designated {\em source} node has an infinite sequence of packets to spread to
the entire network, implementing a one-to-all packet stream.
Formally, this version of the capacity problem constrains the flow set to only
contain a single pair of the form
$\{s, V\setminus \{s\}\}$, for some source $s\in V$. As we will show, the achievable throughput for this problem in a given network graph $G$
is strongly related to $d(G)$, the maximum degree of the minimum degree spanning tree (MDST)
for $G$ (see Section~\ref{sec:prelim}). 
\iflong
We begin in Section~\ref{sec:bcast:upper} by proving that in an arbitrary graph $G$, the achievable throughput is at most $O(1/d(G))$.
This result leverages the classical connection between graph {\em toughness} and spanning
trees discussed and extended in \iflong Section~\ref{sec:prelim}. \fi \ifshort the full version. \fi
We then prove in Section~\ref{sec:bcast:lower} that 
this bound is nearly tight by describing a distributed algorithm
that achieves throughput in $\Omega(1/d(G))$.

We conclude in Section~\ref{sec:bcast:random} by proving that with high probability,
our algorithm from Section~\ref{sec:bcast:lower} will achieve
{\em constant} throughput in a network generated randomly by the $GK$ process.
This indicates that the mobile telephone model is well-suited for this variation of 
the capacity problem.

\fi

\subsection{A Bound on Achievable Throughput for Arbitrary Networks}
\label{sec:bcast:upper}

We establish that the maximum degree of an MDST in $G$---that is, $d(G)$---bounds
the achievable throughput, with larger values of $d(G)$ leading to lower throughput.
The bound is primarily graph theoretic: arguing a fundamental limit on the rate
at which packets can spread through a given topology. 

\iflong
The intuition for this result is the following.
Let $T$ be an MDST in $G$ of degree $d(G)$.
Theorem~\ref{thm:toughness} tells us that there exists a set $S$ of {\em bridge}
nodes such that removing $S$ partitions the graphs into a set $C$ of at least $(d(G) - 2)\cdot |S|$ components.
To spread a given token to all nodes requires that it spread to all components in $C$.

Because removing $S$ creates these partitions, this spreading must pass through nodes in $S$ to conclude.
Because connections are pairwise, however, each bridge node can serve at most one component per round.
Since $|C| \geq (d(G) - 2)\cdot |S|$, it will thus require $\Omega(d(G))$ rounds to complete such a spread.
Hence the latency of the information spreading is lower bounded by the maximum degree of the MDST.  Moreover, this process does not benefit from pipelining: broadcasting $k$ packets will take $\Omega(k \cdot d(G))$ rounds.  This argument is formalized in the next theorem.

\fi

\begin{theorem}
Fix a connected network graph $G=(V,E)$ and broadcast flow set $F$ with source $s$.  Then every schedule achieves throughput at most $O(1/d(G))$.  
\label{thm:bcast:lower}
\end{theorem}
\begin{proof}
Fix some $G=(V,E)$,  $s\in V$, and ${\cal A}$, as specified by the theorem statement.
If $d(G) \leq 4$ then the theorem is trivially true as all throughput values are in $O(1)$.
Assume therefore that $d(G) > 4$. This allows us to apply Theorem~\ref{thm:toughness} for $k=d(G)-1$,
which establishes that
there exists a non-empty subset $S\subset V$ such that $c(G\setminus S) > q\cdot |S|$,
for $q=k-2=d(G)-3>1$ (where, as defined in Section~\ref{sec:prelim}, $c(G \setminus S)$ is
the number of connected components after removing nodes in $S$ from graph $G$).

Let $C$ be the set of components in $G\setminus S$ that do not include the source $s$.
Fix a packet $t$ spread by $s$.
We say $t$ {\em arrives} at  $C_i \in C$ in round $r \geq 1$,
if this is the first round in which a node in $C_i$ receives packet $t$.
In this case, $t$ must have been previously received
by some bridge node in $S$ that is adjacent to $C_i$.
This holds because if $t$ can make it from $s$'s component to $C_i$ without passing through
a node in $S$, then removing $S$ would not  disconnect $C_i$.

Fix any packet count $i\geq 1$.
Each packet requires $|C| = c(G\setminus S) -1 \ge q|S|$ arrival events before it completes spreading.
As we established above, each arrival event requires a given node in $S$ to receive
the given packet. Because each node in $S$ can receive at most one packet per
round, there are at most $|S|$ arrival events per round in the network.

Putting together these pieces,
let $T_i$ be the number of rounds required to spread $i$ packets.
We can lower bound this value as:
\[ T_i \ge \frac{i\cdot |C|}{|S|} = \frac{i(q|S|)}{|S|} 
= iq\ . \]
%
\noindent It follows that for every schedule, 
and every $i$,
at least $T_i$ rounds are required to spread $i$ packets---yielding
a throughput upper bounded by $\frac{i}{T_i} \leq \frac{i}{i\cdot q} = 1/q = 1/(d(G)-3)$, which yields the theorem.
\end{proof}

\subsection{An Optimal Routing Algorithm for Arbitrary Networks}
\label{sec:bcast:lower}

Here we describe a routing algorithm that achieves broadcast capacity throughput in \linebreak$\Omega(1/d(G))$,
when executed in a connected graph $G$.  The high-level idea is to first construct an MDST $T$ in the graph $G$.  We then edge color $T$ using $O(d(G))$ colors, and use this coloring to simulate
the standard CONGEST model, parameterized so that a constant number of packets can fit within its bandwidth limit.
We analyze a straightforward pipelining flooding algorithm for the CONGEST model that converges to constant
throughput. When combined with our simulator, which requires $O(d(G))$ real rounds to simulate each CONGEST
round, the result is a solution that achieves an average latency of $O(d(G))$ rounds per packet,
providing the claimed $\Omega(1/d(G))$ throughput.

As in the pairwise setting, we can do this in a centralized fashion at the cost of a large convergence time (in particular, it takes up to $O(n^2)$ rounds to gather the graph topology locally before we can run a centralized algorithm).    
In order to decrease the convergence time, we \iflong  subsequently \fi describe a distributed
version of this strategy that still converges
to an optimal $\Omega(1/d(G))$ throughput in $O(n^2)$ rounds,
but guarantees to converge to at least $\Omega\left(\frac{1}{d(G) + \log{n}}\right)$ throughput
in  $\tilde{O}(D(T)\cdot d(G) + \sqrt{n})$ rounds,
where $D(T) \leq n$ is the diameter of a spanning tree $T$ built by the algorithm and $\tilde O(\cdot)$ suppresses polylog$(n)$ factors. 


\iflong
\paragraph{Edge Coloring in the Mobile Telephone Model.}
We begin by formally defining an edge coloring:

\begin{definition}
Fix an undirected graph $G=(V,E)$ and palette size $c\geq q$.  A {\em $c$-edge coloring} of $G$
is a function $\pi: E \rightarrow c$ that satisfies the following: 
if $e_1,e_2\in E$ are adjacent then $\pi(e_1) \neq \pi(e_2)$.
\end{definition}

Let $\Delta$ be the maximum degree of $G$. 
Clearly, an edge coloring requires at least $\Delta$ colors.
Vizing showed that this trivial bound is close to optimal
by proving that every graph admits a $(\Delta+1)$-coloring~\cite{vizing}.

Achieving a $(\Delta+1)$-coloring with a centralized algorithm is straightforward.
Because we will also consider distributed broadcast algorithms, however,
we must also discuss how to produce an efficient edge coloring in a distributed manner
in the mobile telephone model.

One of the first {\em distributed} edge coloring algorithms is described
in Luby's seminal paper on the maximal independent set (MIS) problem~\cite{luby}.
He produces a $(2\Delta - 1)$-edge coloring of a graph in $O(\log{n})$ rounds, with high probability,
in the LOCAL model of distributed computing by performing a $(\Delta+1)$-vertex coloring
of the line graph of $G$.
%

We cannot, however, directly run this (or related) distributed coloring strategies in the mobile telephone
model as their efficient time complexities heavily leverage the property of the LOCAL
model that allows unbounded message sizes.\footnote{In Luby's vertex coloring subroutine, for example, each node is responsible for simulating $\Theta(\Delta)$ virtual nodes during an execution of an MIS algorithm. Each node must send an MIS message on behalf of each of its virtual nodes, requiring at least $\Omega(\Delta)$ bits. The message size grows larger when this strategy is applied to the line graph of the original graph. More recent solutions require, at the very least, that nodes frequently describe their current palette of used or unused colors, which also requires $\Omega(\Delta)$ bits.}
Our distributed broadcast capacity algorithm
will need to execute the distributed coloring using the $O(\log{n})$-bit advertisement
tags allowed by our model---a challenge that is equivalent to edge coloring in the broadcast-CONGEST setting with
an $O(\log{n})$ bandwidth limit.
Each broadcast message is therefore only large enough to describe a constant number of colors and/or nodes. 

\begin{figure}
  \begin{center}
   \begin{tabular}{|l|}
   \hline
   {\bf EdgeColor-MTM$(\Delta)$}\\
   \hline
    {\bf for} $i \gets 1$ to $2\Delta-1$ \\
    $\ \ \ \ \ $ {\bf construct} a maximal matching $M$ using II\\
     $\ \ \ \ \ $ {\bf color} edges in $M$ with color $i$  \\
     $\ \ \ \ \ \ \ \ $  (i.e., for each $e\in M$ set $\pi(e) \gets i$) \\
      $\ \ \ \ \ $ {\bf remove} edges in $M$\\
   {\bf return} $\pi$ \\
   \hline
   \end{tabular}
     \end{center}
  \caption{Edge coloring strategy for the mobile telephone model. Notice all actual distributed
  coordination occurs during the maximal matching step which uses the II algorithm
  due to Israeli and Itai~\cite{israeli}.}
  \label{fig:color}
\end{figure}

Notice that the small broadcast messages in our setting implies that a $\Omega(\Delta)$ bound is unavoidable,
as these many rounds are required for even basic coloring activities like describing your current used/unused palette,
or assigning a color to each neighbor. 
On the positive side, 
the necessity of a slower bound enables us to explore simpler solutions.
In particular, we propose the strategy summarized in Figure~\ref{fig:color},
in which nodes repeatedly construct a maximal matching, 
coloring the edges in the current matching with a new color and then removing them from consideration for future
matchings.

As we establish, this strategy always terminates in at most $2\Delta-1$ matchings (creating a palette of the same 
size), and the $O(\log{n})$-round maximal matching algorithm of Israeli and Itai~\cite{israeli} is easily adapted
to work with the small advertisement tags in our model (it requires nodes to broadcast,
at most, a constant number of identifiers per round). The result is a randomized
$(2\Delta-1)$-edge coloring algorithm that works in $O(\Delta\log{n})$ rounds, with high probability.\footnote{The high probability
in the maximal matching algorithm is on the time complexity. Formally, we run the maximal matching algorithm for a fixed duration
of rounds. With low probability, these rounds are not enough for one of the maximal matchings to succeed,
potentially resulting in an incomplete edge coloring. Our edge coloring strategy is therefore
a Monte Carlo algorithm, which will simplify its later use as a subroutine in a larger distributed system.}
Formally:

\begin{theorem} \label{thm:edge-color}
The {\em EdgeColor-MTM$(\Delta)$} algorithm produces a $(2\Delta-1)$-edge coloring in $O(\Delta\log{n})$
rounds, with high probability in $n$.
\end{theorem}
\begin{proof}
It follows directly from the definition of the algorithm that no two adjacent edges are colored the same color,
as this would require two adjacent edges to be included in the same matching.
It is sufficient, therefore, to show that $2\Delta-1$ maximal matchings are sufficient to cover every edge in $E$.

To see why this is true, fix some edge $(u,v)$ in $G$.
The only event that can prevent $(u,v)$ from being included in a given maximal matching is if at least one other
edge adjacent to $u$ or $v$ is included in the matching.
There are at most $2(\Delta-1)$ such other edges,
so $(u,v)$ must be matched after at most $2(\Delta-1)+1 = 2\Delta -1$ matchings.

The high probability comes from the Israeli and Itai maximal matching algorithm,
which always produces a maximal matching, but terminates in $O(\log{n})$ rounds with high probability.
By a union bound,  with high probability all $2\Delta-1 = O(n)$ instances of the algorithm
terminate in time. 
\end{proof}

\begin{figure}
  \begin{center}
   \begin{tabular}{|l|}
   \hline
   {\bf SB($s$)}\\
   \hline
   {\bf construct} an approximate MDST $T$ \\
   %
   %
   {use} $T$ to {\bf convergecast} and {\bf broadcast} the max degree $d$ of $T$ \\
   %
   {\bf color} the edges in $T$ using {\em EdgeColor-MTM$(d)$}\\
   (check validity of coloring and repeat if problem found)\\
   {\bf use} edge colors to simulate the following CONGEST strategy:\\
    $\ \ \ \ $ {\bf pipeline} message floods from $s$ in $T$\\
   \hline
   \end{tabular}
     \end{center}
  \caption{Broadcast strategy for a given source node $s$. The final step requires
  nodes to flood messages to their children, participating in a new flood in each round,
  which requires a simulation of the CONGEST model
  in which the number of connections at each node in each round is unrestricted.}
  \label{fig:bcast}
\end{figure}

\paragraph{A Tight Broadcast Capacity Algorithm.} 
In Figure~\ref{fig:bcast},
we describe our streaming broadcast strategy: we create an MDST, edge color it, and then use this edge coloring to simulate pipelined flooding.  
In a centralized setting, we use the best-known approximation for MDST~\cite{fr94}, which gives a spanning tree $T$ with max degree $d(G) + 1$.  We also use the centralized edge-coloring strategy for $T$ that uses only $d + 1$ colors (which is easy on trees).  

In the distributed setting, we assume the MDST algorithm is a Las Vegas algorithm that terminates with a tree with a degree in $O(d(G))$
in $f(n)$ rounds, for some complexity function $f(n)$ that we discuss below,
with high probability in $n$ (and with probability $1$ in the limit).  We also use the edge coloring algorithm from above as a subroutine.
As described earlier, {\em EdgeColor-MTM} is a Monte Carlo algorithm.
With low probability it can fail to color all edges.
The validity checking step of our streaming broadcast algorithm simply checks for failures with a convergecast
on the tree. If no problems are reported, $s$ can broadcast a message telling the network to proceed.

Once the edges are colored with $c=2d-1$ colors (or $c = d+1$ colors for the centralized algorithm),
we can easily simulate the CONGEST model using $c$ real rounds for each simulated round.  To simulate one round of CONGEST, we just cycle through the the $c$ colors, allowing, for each color $i$, 
all edges colored $i$ to connect. By the definition of edge coloring, all edges with the same color form a matching, so in $c$ rounds we can simulate one round of CONGEST (where every node sends a message to \emph{all} of its neighbors rather than just one).  

The pipeline flood we run on this simulation is the simple strategy in which $s$ floods messages
down the MDST tree, starting a new flood in each round.
The result is a pipeline of floods in which nodes receive a new message in every round.
Because each simulated round requires $c = \Theta(d)$
rounds, and $d$ is the maximum degree of an approximate MDST,
the result is a throughput that converges to $\Omega(1/d(G))$, as needed.
Formally:

\begin{theorem}
When executed in a connected network topology $G=(V,E)$ of size $n=|V|$, with a broadcast capacity flow set with
source $s\in V$,
with high probability in $n$: the {\em SB$(s)$} algorithm achieves a throughput in $\Omega(1/d(G))$ with 
respect to $G$ and $F$.
\label{thm:bcast}
\end{theorem}
\begin{proof}
With high probability in $n$, we successful construct an MDST with maximum degree in $O(d(G))$,
and successfully color the edges with $O(d(G))$ colors.
Once this initialization is complete, the claimed throughput is achieved
once sufficient packet deliveries have passed to amortize the setup costs of the MDST, edge coloring, and
pipeline initialization.
\end{proof}

\paragraph{Convergence Time.}
To understand the convergence round of $SB(s)$ (that is, how fast it converges
to its claimed throughput), we must consider the three setup costs the algorithm
pays before converging to its eventual throughput: (1) MDST setup; (2) edge coloring;
(3) time required to fill the pipeline.
Tackling these in reverse order, the third requires $D(T)$ CONGEST rounds 
(where $T$ is the MDST tree built by the algorithm), which works out to $O(D(T)\cdot d(G))$ real rounds,
and (2) is $O(\Delta(T)\log{n})$ rounds with high probability by Theorem~\ref{thm:edge-color} (where $\Delta(T)$ is the maximum degree
of $T$). We note that $d(G) = \tilde \Theta(\Delta(T))$, so the third cost dominates the second cost when ignoring log factors.

The cost of (1) depends on the algorithm deployed.
The centralized MDST algorithm due to F{\"u}rer and Raghavachari~\cite{fr94}, generates a tree with maximum
degree $d(G) + 1$ in a polynomial number of computational steps. 
In our distributed setting,
this algorithm can be deployed by spending $O(n^2)$ rounds to gather
the entire network topology by flooding edge descriptions using the mobile telephone model advertisements
(which can fit a constant number of edges per advertisement), 
and then have each node run the centralized algorithm locally.

In recent work, Dinitz et~al.~\cite{dinitz:2019} present a distributed algorithm for the broadcast-CONGEST
model (with bandwidth bound $O(\log{n})$), 
that constructs a tree with maximum degree $O(d(G) + \log{n})$ in $\tilde{O}(D + \sqrt{n})$ rounds,
with high probability in $n$. We can directly run this algorithm in the mobile telephone model
using the advertisements to implement the small broadcast messages from the broadcast-CONGEST model.
This distributed solution is more efficient, but for networks with small $d(G)$ values, it does not enable
optimal throughput.

As hinted earlier, we can balance these competing interests by combining the two algorithms.
In particular, we can implement $SB(s)$ such that it begins by constructing a tree using
the distributed algorithm from~\cite{dinitz:2019}. It can then, in the background,
improve this tree down to degree $d(G) + 1$ using the algorithm from~\cite{fr94},
switching to the new tree once it is complete.

By combining these various costs, we converge to $\Omega(1/(d(G) + \log{n}))$ throughput
in $\tilde{O}(D(T)\cdot d(G) + \sqrt{n})$ rounds,
which then improves to $\Omega((1/d(G)))$ throughput within $O(n^2)$ rounds. 

\fi

\subsection{Random Networks}
\label{sec:bcast:random}

The preceding broadcast capacity results hold for any connected network graph.
Here we study the problem in networks randomly generated by the $GK$ process 
with a communication radius sufficiently larger than the threshold $r_c(n)$.


\iflong
In more detail, we prove that for any radius that is a sufficiently large constant factor bigger
than $r_c(n)$, with high probability in $n$,
our $SB(s)$ routing algorithm from Section~\ref{sec:bcast:lower} will achieve
{\em constant} throughput in a network generated by $GK(n, r)$---indicating 
that in a natural network topology, the mobile telephone model is well-suited for broadcast
capacity.
The proof of the below theorem leverages results from Section~\ref{sec:gk:lower}
to prove that the network likely has a constant degree MDST. 
\fi

\iflong
\begin{theorem}
There exists a constant $\beta \geq 1$,
such that for any significantly large network size $n>1$
and radius $r \geq \beta r_c(n)$,
if $G \sim GK(n,r)$ then with high probability $SB(s)$ achieves constant throughput (for any $s$).
\label{thm:bcast:random}
\end{theorem}
\begin{proof}
By Lemma~\ref{lem:gk:good}, there exists a constant $\beta$,
such that for any sufficiently large network size $n$,
and any radius $r \geq \beta\cdot r_c(n)$,
the graph $G=(V,E)$ generated by the $GK(n,r)$ process is good (see Definition~\ref{def:good}).
Assume this holds.
Consider the $r$-grid. 
By the definition of good and $r$-grid,
each grid box is non-empty, and
each node is within range of every node in grid boxes that share an edge with its own.

With this in mind, fix one {\em core} node in each grid box. To construct a spanning tree, first
connect each chosen core node
to the chosen core nodes in the (at most) four adjacent boxes.
This creates an overlay with at least one node in every grid box.  We then take an arbitrary spanning tree of this overlay.  Finally, for each box, connect the other nodes into a line that includes the box's chosen core node as its endpoint.
The result is a spanning tree with maximum degree $5$.
It follows that $d(G) \leq 5$.
By Theorem~\ref{thm:bcast}, the $SB(s)$ algorithm will achieve throughput in 
$\Omega(1/(d(G)) = \Theta(1)$ in this graph.
\end{proof}

\fi

\section{All-to-All Capacity}
\label{sec:all}

We now consider the all-to-all capacity problem,
which assumes all nodes begin with an infinite sequence of packets
to spread to all other nodes.
Formally, this variation of the capacity problem considers only the 
following canonical flow
set:  $F_{all} = \{ (s, V \setminus \{s\}) : s\in V \}$.

In Section~\ref{sec:bcast},
we described and analyzed an algorithm that achieved a throughput
in $\Omega(1/d(G))$ for delivering packets from a single source to the whole network.
To solve all-to-all capacity, we could run $n$ instances of this algorithm:
one for each source, rotating through the different instances in a round robin fashion.
This approach provides a baseline throughput result of $\Omega(1/(n\cdot d(G)))$.
The key questions are whether or not this bound is tight,
and whether there are simpler or more natural strategies than deploying round robin interleaving of single-source
broadcast.

\iflong
In this section we answer both questions in the affirmative.
We first prove that every schedule achieves a throughput in $O(1/(n\cdot d(G)))$,
then describe and analyze a distributed algorithm that achieves optimal throughput
by solving all-to-all gossip with a simple flood on a good spanning tree for each packet.
This algorithm then becomes our basis for largely resolving
an open question from~\cite{newport:2017} regarding
one-shot gossip in the mobile telephone model.

\subsection{A Bound on Achievable Throughput for Arbitrary Networks}
\label{sec:all:upper}

In Section~\ref{sec:bcast}, we proved a tight connection between the achievable broadcast capacity
and the degree of an MDST in the graph (i.e., $d(G)$ for graph $G$).
Here we formalize the intuition that this connection also exists for the related problem
of all-to-all capacity.

\begin{theorem}
Fix a connected network graph $G=(V,E)$ of size $n=|V|$.  Every schedule achieves throughput at most $O\left(\frac{1}{n\cdot d(G)} \right)$ with respect to $G$ and $F_{all}$.
\label{thm:all:upper}
\end{theorem}
\begin{proof}
Fix some $G=(V,E)$ as specified by the theorem statement.
Because a given node in our model can receive at most one packet per round,
it is trivial to calculate that for every packet count $i$,
it requires  $\Omega(n\cdot i)$ rounds for all nodes to deliver their first $i$ packets to all other nodes.
Therefore, if $d(G)$ is constant, the theorem is trivially true.

On the other hand, if $d(G)$ is a sufficiently large
constant, we can apply the same argument from our proof of Theorem~\ref{thm:bcast:lower},
which in turn leverages Theorem~\ref{thm:toughness},
to establish that $T_k = \frac{(d(G)-3) \cdot k}{2}$ rounds
are required for $k$ packets to spread.

Focusing on this case, fix some packet count $i$.
For all $n$ nodes to successfully spread $i$ distinct packets requires $k= n\cdot i$ total packets to spread in the network,
requiring at least $T_{k} \in \Omega(d(G) \cdot n \cdot i)$ total rounds, yielding
a throughput in $O(i/T_k) = O(1/(d(G)\cdot n))$.
Since this holds for all $i$, the theorem claim follows.
\end{proof}

\subsection{An Optimal Routing Algorithm for Arbitrary Networks}
\label{sec:all:lower}

In Section~\ref{sec:bcast:lower},
we described an algorithm that simulates the CONGEST model (with a bandwidth bound sufficient to fit
a constant number of packets) in the mobile telephone model.
The strategy first builds an approximate MDST and then edge colors the tree edges.
This coloring is used to schedule the mobile telephone model connections needed
to simulate on round of CONGEST. If the tree has maximum degree $d$, then the simulation
requires $O(d)$ real rounds for each simulated CONGEST round.

To match the all-to-all capacity bound from Theorem~\ref{thm:all:upper},
we deploy this same CONGEST simulation.
This time, however, we run an all-to-all gossip algorithm on top of the simulation, and show that this algorithm spreads $k$ gossip messages to all nodes
in $O(D + k)$ rounds in the CONGEST model in a network with diameter $D$.
Combining this bound with the simulation overhead will yield a throughput result
that asymptotically matches the bound from Theorem~\ref{thm:all:upper}. 

We begin below by describing and analyzing our CONGEST gossip algorithm,
before analyzing how it combines with our simulator.

\paragraph{Broadcast Gossip in CONGEST.}
Here we describe and analyze a simple strategy call {\em broadcast gossip}, 
that is designed for the CONGEST.
The strategy works as follows. Every node maintains a FIFO message queue initialized to
holds its initial gossip message (if it starts with such a message),
and a list of sent messages initialized to be empty.
At the beginning of each round, each node $v$ does the following.
If its queue is non-empty, it dequeues a message,  broadcasts it to all of its neighbors,
and adds it to its sent messages list.
For each message $m$ received by $v$ in this round,
if $m$ is on $v$'s sent messages list, it discards it, otherwise it enqueues it to its message queue.

The time complexity of this strategy is well-known as folklore,
and for some models, it has a concrete proof in the literature; e.g.,~\cite{mmb-absmac}.
For the sake of completeness, we clarify and generalize the result from~\cite{mmb-absmac}:


\begin{theorem}
Consider the broadcast gossip algorithm used to spread $k\geq 1$ messages in connected network topology $G=(V,E)$
with diameter $D$. All messages spread to all nodes by the end of round $D+k$.
\label{thm:gossip}
\end{theorem}
\begin{proof}
We begin by defining two useful pieces of notation: for a given gossip message $m$, let $u_m$ be 
the node that starts with $m$, and for each $v \in V$, let $d_m(v)$ be the shortest path distance between $v$ and $u_m$.
Fix a specific message $m$.
We first study the spread of this message through the network using induction on the round number.
In particular, consider the following inductive hypothesis: 

\begin{quote}
{\em For every $r\geq 1$: for every $v$ and $\ell$
such that $d_{m}(u) + \ell = r$, one of the following two properties must be true of $v$ after round $r$: 
(1) $v$ has sent $m$; (2) $v$ has sent at least $\ell$ distinct gossip messages. }
\end{quote}

We begin with the base case ($r=1$).
There are only two relevant combinations of $v$ and $\ell$ values for $r=1$.
The first is when $d_{m}(v) = 0$ and $\ell = 1$.
In this case, $v=u_m$ so we know $v$ starts with $m$ and therefore has a message to broadcast during round $1$, satisfying
property (2) for $\ell=1$.
The second case is when $d_{m}(v) = 1$ and $\ell = 0$.
In this case, property (2) is vacuously true. 

We continue with the inductive step ($r>1$). 
Fix any $v$ and $\ell$ such that $d_{m}(v) + \ell = r$.
Let us consider what has happened by the end of round $r-1$.
If $v$ has sent $m$ by the end of round $r-1$, then we are done. 
Moving forward, therefore, assume $v$ has not sent $m$ by the end of $r-1$.

Fix some $w$ that is one hop closer to $u_m$ than $v$.
Notice that $r-1 = d_{m}(v) + \ell - 1 = d_{m}(w) + \ell$.
Therefore, by the inductive hypothesis, and our assumption that $v$ has not yet
sent $m$, we know that after round $r-1$,
node $v$ has broadcast at least $\ell-1$ messages and node $w$ has either broadcast
$m$ or broadcast $\ell$ messages.
Either way, $v$ has at least one new message to broadcast in $r$,
meaning that by the end of this round it will have at least satisfied property (2) of the inductive hypothesis. 

Stepping back, we can now pull together the pieces to prove the main theorem.
The inductive claim above establishes that for each $v$,
$v$ has sent $m$ by round $d_m(v) + k \leq D + k$.
This follows because by the above hypothesis,
by the end of round $d_m(v) + k$,
$v$ has either sent $m$ or at least $k$ other messages.
Given that there are only $k$ total messages, the latter property also implies it has sent $m$.

The inductive claim applies for every message $m$.
Therefore, every node has sent (and therefore received) every message by round $D+k$, as claimed.
\end{proof}

\begin{figure}
  \begin{center}
   \begin{tabular}{|l|}
   \hline
   {\bf SG}\\
   \hline
   {\bf construct} an approximate MDST $T$ \\
   %
   %
   {use} $T$ to {\bf convergecast} and {\bf broadcast} the max degree $d$ of $T$ \\
   {\bf color} the edges in $T$ using {\em EdgeColor-MTM$(d)$}\\
   (check validity of coloring and repeat if problem found)\\
   {\bf for each token} $i = 1, 2, 3, ...$: \\
   $\ \ \ \ $ {\bf use} edge colors to simulate {\bf broadcast gossip} for token $i$. \\
   \hline
   \end{tabular}
     \end{center}
  \caption{Streaming gossip strategy. We proved earlier that the broadcast gossip algorithm strategy 
  terminates in at most $\alpha n$ simulated rounds, for a fixed constant $\alpha$.
  Nodes can therefore simulate run each simulation of this strategy for this fixed number
  of simulated rounds before moving on to the next token.}
  \label{fig:all}
  \end{figure}

\paragraph{A Tight All-to-All Capacity Algorithm.}
As with broadcast capacity, we build an approximate MDST, edge color it, and then
use the colors to simulate CONGEST.
For each packet count $i$,
we use our simulation to 
 run the broadcast gossip strategy analyzed above to gossip each node's packet number $i$.
 If $\hat d$ is the maximum degree of the approximate MDST,
 and $\hat D$ is its diameter,
 then each instance of broadcast gossip requires $O(n + \hat D) = O(n)$ simulated
 rounds, which in turn requires $O(n\cdot \hat d)$ total rounds---providing throughput values
 that match the $O(1/(d(G)\cdot n))$ bound proved in Theorem~\ref{thm:all:upper}.
 Formally:

\begin{theorem}
When executed in a connected network topology $G=(V,E)$ of size $n=|V|$, 
with high probability in $n$: the {\em SG} algorithm achieves throughput in $\Omega\left(\frac{1}{d(G)\cdot n}\right)$ with 
respect to $G$ and $F_{all}$.
\label{thm:all}
\end{theorem}
\begin{proof}
Fix some $G=(V,E)$ as specified by the theorem.
With high probability in $n$,
we successfully setup a simulation of the CONGEST model in a tree with a maximum degree in $O(d(G))$.
The {\em SG} algorithm solves all-to-all gossip for each packet, finishing the current packet
before moving onto the next.
By Theorem~\ref{thm:gossip},
each instance requires $O(D+n) = O(n)$ simulated rounds,
which requires $O(n\cdot d(G))$ real rounds.
For all packet counts, the throughput is therefore in $\Omega\left(\frac{1}{d(G)\cdot n}\right)$, as claimed.
\end{proof}

When it comes to convergence time,
the same arguments as in Section~\ref{sec:bcast:lower} apply.
That is, we can use a hybrid of the approximate MDST algorithms from~\cite{dinitz:2019,fr94}
to efficiently achieve a throughput in $\Omega\left(\frac{1}{(d(G) + \log{n})\cdot n}\right)$
that then improves to $\Omega\left(\frac{1}{d(G)\cdot n}\right)$ rounds by $O(n^2)$ rounds.
(In this case, the definition of {\em efficient} is slightly improved
as compared to broadcast
capacity, as we solve all-to-all gossip from scratch for each packet, eliminating a relevant setup cost
related to filling a pipeline.)

\subsection{Random Networks}
\label{sec:all:random}

We now prove that as with broadcast capacity,
randomly generated networks are likely to enable efficient packet spreading.
In particular,
we prove that for any sufficiently large network size $n$,
and any sufficiently large radius $r$ compared to the connectivity threshold $r_c(n)$ (see Section~\ref{sec:bcast:random}),
with high probability in $n$,
our $SG$ algorithm will achieve a throughput in $\Omega(1/n)$ with respect a graph generated
by the $GK(n,r)$ process and flow set $F_{all}$.

Notice, because a node can receive at most one new packet per round, $O(1/n)$ is a trivial bound
on achievable throughput in {\em every} graph, so this shows that graphs from the $GK$ process are in some sense the easiest graphs.  
The below result follows directly from Theorem~\ref{thm:all} and the argument
used in the proof of Theorem~\ref{thm:bcast:random} that establishes
 for a sufficiently large radius, the resulting graph is likely to have a constant degree spanning tree.

\begin{theorem}
There exists a constant $\beta \geq 1$,
such that for any significantly large network size $n>1$
and radius $r \geq \beta r_c(n)$,
if $G \sim GK(n,r)$ then with high probability 
$SG$ achieves throughput in $\Omega(1/n)$ with respect to $G$ and $F_{all}$.
\label{thm:all:random}
\end{theorem}

\fi

\subsection{Implications for One-Shot Gossip}
\label{sec:all:oneshpt}

Existing results for one-shot gossip in the mobile telephone model
are expressed with respect to the vertex expansion (denoted $\alpha$) of the graph topology~\cite{newport:2019,newport:2017}.
The best known results requires $O((n/\alpha)\text{polylog}(n))$ rounds,
which is not tight in all graphs as vertex expansion does not necessarily characterize optimal gossip.\footnote{Consider, for example, a path of length $n$, which has $\alpha = 2/n$. It is possible to pipeline $n$ messages through this network
in $\Theta(n)$ rounds, which is much faster than $\tilde{O}(n/\alpha) = \tilde{O}(n^2)$.}
A key open question from~\cite{newport:2017} is whether it is possible to produce a gossip algorithm
that is optimal (or within log factors of optimal) in {\em all} network topology graphs.
The techniques used in the above capacity bounds help us prove the following,
which largely resolves this open question:

\begin{theorem}
Fix a connected network topology $G = (V,E)$ with diameter $D$, size $n=|V|$,
and MDST degree $d(G)$.
Every solution to the one-shot gossip in $G$ requires $\Omega(d(G)\cdot n)$ rounds.
There exists an algorithm solves the problem in
$O((D+\sqrt{n})\text{\em polylog}(n) + n(d(G) + \log{n}))= \tilde{O}(d(G)\cdot n)$ rounds, with high probability in $n$. 
\label{thm:oneshot}
\end{theorem}
\iflong
\begin{proof}
The lower bound of $\Omega(d(G)\cdot n)$ follows directly
from the argument of Theorem~\ref{thm:all:upper}.
To derive the upper bound,
we consider the $SG$ algorithm run with the distributed approximate MDST
algorithm from~\cite{dinitz:2019}.
After initialization,
$SG$ solves one-shot gossip in $O((D+n)\cdot \hat d) = O(n\cdot \hat d)$ rounds,
where $\hat d$ is the maximum degree of the tree used by the algorithm.
Using the distributed algorithm from~\cite{dinitz:2019},
this tree requires $\tilde{O}(D + \sqrt{n})$ rounds
to construct and yields $\hat d \in O(d(G) + \log{n})$.
The theorem claim follows directly.
\end{proof}

Notice,
the solution described in Theorem~\ref{thm:oneshot}
is asymptotically optimal in any graph with $d(G) \in \Omega(\log{n})$ and $D \in O(n/\log^x{n})$ (where $x$
is the constant from the polylog in the MDST construction time), which describes a large family of graphs.
Furthermore, for the subset of graphs with small MDST degrees and/or large degrees,
the solution can be expressed as $\tilde{O}(d(G)\cdot n)$, 
which is at most a polylog factor slower than optimal.

This is the first known gossip solution to be optimal, or within log factors of optimal, in {\em all} 
graphs, largely answering the challenge presented by~\cite{newport:2017}.

\fi



\bibliographystyle{plainurl}
\bibliography{main}

\newpage
\appendix

\end{document}